\documentclass[11pt]{article}

\usepackage{natbib}
\usepackage{makecell} %
\usepackage{booktabs}
\usepackage{amsmath,amsthm,amssymb} %
\usepackage[affil-it]{authblk} %
\usepackage[english]{babel} %
\usepackage{caption} %
\usepackage{color} %
\usepackage{algorithmic} %
\usepackage{algorithm} %
\usepackage{enumitem} %
\usepackage{mathpazo} 
\usepackage{framed} %
\usepackage[margin=1in]{geometry} %
\usepackage{graphics} %
\usepackage{hyphenat} %
\usepackage[breaklinks]{hyperref} %
\usepackage{mathabx} %
\usepackage{mathtools} %
\usepackage{microtype} %
\usepackage[utf8]{inputenc} %
\usepackage{soul} %
\usepackage{subcaption} %
\usepackage{subdepth} %
\usepackage{suffix} 
\usepackage{tikz} %
\usepackage{xspace} %
\usepackage{thmtools}
\usepackage{thm-restate}
\usepackage{interval}
\intervalconfig{soft open fences}
\usepackage{cleveref}

\hypersetup{colorlinks=true, linkcolor=blue, citecolor=magenta} %
\definecolor{shadecolor}{rgb}{.95,.95,.95} %

\setul{1ex}{.5pt} %
\setlist[enumerate]{nolistsep,itemsep=3pt,topsep=3pt} %

\definecolor{White}{rgb}{1,1,1} %
\definecolor{Black}{rgb}{0,0,0} %
\definecolor{LightGray}{rgb}{.8,.8,.8} %

\newtheorem{theorem}{Theorem} %
\newtheorem{lemma}[theorem]{Lemma} %
\newtheorem{proposition}[theorem]{Proposition} %
\newtheorem{corollary}[theorem]{Corollary} %
\newtheorem{definition}[theorem]{Definition} %
\newtheorem{fact}[theorem]{Fact} %

\newcommand{\microspace}{\mspace{.5mu}} %

\newcommand{\norm}[1]{\left\lVert #1 \right\rVert}
\newcommand{\norms}[1]{\lVert #1 \rVert}

\newcommand{\ips}[2]{\langle #1,#2 \rangle}

\newcommand{\floors}[1]{\lfloor #1 \rfloor}
\newcommand{\ceil}[1]{\left\lceil #1 \right\rceil}
\newcommand{\ceils}[1]{\lceil #1 \rceil}
\newcommand{\abss}[1]{\left\lvert #1 \right\rvert}
\newcommand{\abs}[1]{\lvert #1 \rvert}
\newcommand{\paren}[1]{\left( #1 \right)}
\newcommand{\parens}[1]{( #1 )}
\newcommand{\sqb}[1]{\left[ #1 \right]}
\newcommand{\sqbs}[1]{[\![ #1 ]\!]}
\newcommand{\set}[1]{\left\{ #1 \right\}}
\newcommand{\sets}[1]{\{ #1 \}}
\newcommand{\ket}[1]{\ensuremath{\lvert\microspace #1
    \microspace\rangle}} %
\newcommand{\RR}{\mathbb{R}} %
\newcommand{\ZZ}{\mathbb{Z}} %
\newcommand{\diag}{\operatorname{diag}} %
\newcommand{\poly}{\operatorname{poly}} %
\newcommand{\multiscale}{\operatorname{ms}} %
\newcommand{\trace}{\operatorname{Tr}} %
\newcommand{\rank}{\operatorname{rank}} %
\renewcommand{\epsilon}{\varepsilon} %

\def\cA{\mathcal{A}} 
\def\cJ{\mathcal{J}} 
\def\cF{\mathcal{F}} 
\def\cW{\mathcal{W}} 
\def\cO{\mathcal{O}} 
\def\cL{\mathcal{L}} 
\def\cZ{\mathcal{Z}} 

\begin{document}

\title{\Large\bf Accelerating Regression Tasks with Quantum Algorithms}
\renewcommand*{\Affilfont}{\small\itshape} 
\author[1,2]{Chenghua Liu}
\author[3]{Zhengfeng Ji}
\affil[1]{
  Institute of Software, Chinese Academy of Sciences, Beijing, China
}
\affil[2]{
  University of Chinese Academy of Sciences, Beijing, China
}
\affil[3]{
Department of Computer Science and Technology, Tsinghua University, Beijing, China
}

\maketitle

\begin{abstract}
  Regression is a cornerstone of statistics and machine learning, with
  applications spanning science, engineering, and economics.
  While quantum algorithms for regression have attracted considerable attention,
  most existing work has focused on linear regression, leaving many more complex
  yet practically important variants unexplored.
  In this work, we present a unified quantum framework for accelerating a broad
  class of regression tasks---including linear and multiple regression, Lasso,
  Ridge, Huber, $\ell_p$-, and $\delta_p$-type regressions---achieving up to a
  quadratic improvement in the number of samples $m$ over the best classical
  algorithms.
  This speedup is achieved by extending the recent classical breakthrough of
  \citet{jambulapati2024sparsifying} using several quantum techniques, including
  quantum leverage score approximation~\citep{apers2023quantum} and the
  preparation of many copies of a quantum state~\citep{hamoudi2022preparing}.
  For problems of dimension $n$, sparsity $r < n$, and error parameter
  $\epsilon$, our algorithm solves the problem in
  $\widetilde{O}(r\sqrt{mn}/\epsilon + \mathrm{poly}(n,1/\epsilon))$\footnote{We
    use $\widetilde O\parens{f}$ to represent
    $O \paren{f \cdot \poly\log\parens{m, n,1/\epsilon}}$ throughout this paper
    to suppress polylogarithmic factors.}
  quantum time, demonstrating both the applicability and the efficiency of
  quantum computing in accelerating regression tasks.
\end{abstract}

\section{Introduction}
Regression lies at the heart of statistical modeling and data analysis, providing a fundamental framework for understanding relationships between variables and making predictions.
 Among its various forms, linear regression stands out as one of the most widely used and longest-standing tools, with a history dating back over three centuries and applications ranging from the natural sciences to economics and engineering~\citep{legendre1806nouvelles, fisher1970statistical, neter1996applied}.
Over the decades, the basic linear regression model has been extended in multiple directions, including  ridge and lasso regression, $\ell_p$ regression, Huber regression, each addressing different statistical and computational challenges. 
 These models have become indispensable in modern data analysis, driving extensive theoretical research on fast regression methods. Among them, $\ell_p$ regression has witnessed a series of recent breakthroughs~\citep{adil2019iterative, lee2019solving, musco2022active, adil2024fast, jambulapati2024sparsifying}, alongside many other progress in other regression tasks.

 With the rapid progress of quantum hardware and algorithms, there has been growing interest in leveraging quantum computing to accelerate optimization and machine learning tasks. Among them, regression stands out as a core problem in statistics and machine learning, and has attracted significant attention from the quantum algorithms community.
 A rich body of work has explored quantum approaches to linear regression under various models and assumptions, ranging from HHL-based solvers and conditional gradient method to quantum leverage score techniques~\citep{wiebe2012quantum, schuld2016prediction,
  wang2017quantum,
   kerenidis2017quantum,
    chakraborty2019power,
     shao2020quantum, 
     chen2023quantum,
     shao2023quantum,
      chen2023faster,
       chakraborty2023quantum, 
       song2023revisiting}. 
However, for broader and both theoretically and practically important regression tasks such as $\ell_p$ regression and Huber regression, it remains unclear whether similar quantum speedups can be achieved. Given the diversity of regression tasks and their substantial significance in both machine learning and theoretical computer science, a natural question arises:
\begin{center}
\textit{To what extent can the advantages of quantum computing be extended beyond linear regression to encompass this broader family of models?}
\end{center}

In this work, we study this question by adopting the generalized linear model (GLM) sparsification framework and developing a quantum algorithmic approach tailored to it. Sparsification refers to the process of replacing a large dataset with a much smaller weighted subset (a sparsifier) that approximately preserves the objective value for all candidate solutions, thereby allowing the original optimization problem to be solved more efficiently. Once such a sparsifier is constructed, computations performed on it yield an approximate solution to the original problem. Consequently, our framework provides quantum speedups for approximately solving a broad class of empirical risk minimization (ERM) problems, covering linear and multiple regression, Lasso, ridge, Huber, $\ell_p$-, and $\delta_p$-type regressions. This demonstrates both the unifying nature and the broad applicability of our approach in accelerating fundamental regression tasks.

To clearly specify the scope of our framework, we formalize the problem of optimizing a generalized linear model as
\begin{equation}\label{eq:glm}
\min_{x \in \mathbb{R}^n} F(x)\quad \text{for}\quad F(x) := \sum_{i=1}^m f_i\big(\langle a_i, x\rangle - b_i\big),
\end{equation}
where the total loss $F:\mathbb{R}^n \to \mathbb{R}$, is defined by data vectors $a_1,\ldots,a_m \in \mathbb{R}^n$,  a response vector $b \in \mathbb{R}^m$, and loss functions $f_1,\ldots,f_m:\mathbb{R} \to \mathbb{R}_+$.This general formulation serves as a unifying abstraction for many regression tasks, from linear regression to more sophisticated robust objectives.

When the number of samples $m$ is much larger than the dimension $n$, a natural strategy for accelerating computation is to apply sparsification. This technique have been widely used in classical algorithm design, with graph sparsification being one of the most prominent examples~\citep{spielman2011graph,batson2014twice}. In a recent breakthrough, \citet{jambulapati2024sparsifying} systematically studied the GLM sparsification problem and presented a corresponding classical algorithm.  Motivated by this advance, we investigate the potential for quantum speedups in this setting,  leading to our quantum algorithmic GLM sparsification framework. To precisely describe GLM sparsification, we use the following definition:
\begin{definition}[Generalized Linear Model (GLM) Sparsification]
Let the total loss function $F(x) := \sum_{i=1}^m f_i\left(\langle a_i, x \rangle-b_i \right)$ be defined by vectors $a_1, \ldots, a_m \in \mathbb{R}^n$,a vector $b \in \mathbb{R}^m$,  and loss functions $f_1, \ldots, f_m : \mathbb{R} \to \mathbb{R}_+$.
We say that $F$ admits an \emph{$s$-sparse $\epsilon$-approximation over the range $[s_{\min}, s_{\max}]$} if there exist non-negative weights $w \in \mathbb{R}_+^m$, with at most $s$ nonzero entries, such that
\begin{equation*}
\left| F(x) - \widetilde{F}(x) \right| \leq \epsilon \cdot F(x), \quad \forall x \in \RR^n \text{ with }s_{\min}\leq F \parens{x}\leq s_{\max},
\end{equation*}
where $\widetilde{F}(x) := \sum_{i=1}^m w_i f_i\left(\langle a_i, x \rangle-b_i \right)$ is referred to as a \emph{$s$-sparse $\epsilon$-approximate sparsifier of $F$ over the range $[s_{\min}, s_{\max}]$}.
\end{definition}
\subsection{Our results}
We consider a generalized linear model as in \cref{eq:glm}, where the loss functions $f_1,\ldots,f_m : \mathbb{R} \to \mathbb{R}_+$ are applied to  $\langle a_i,x\rangle - b_i$ for $a_i\in\mathbb{R}^n$. Throughout the analysis, we assume for convenience that the bias term is absent, i.e., $b=0$. This assumption entails no loss of generality: any nonzero bias $b_i$ can be incorporated into the vector by appending an additional coordinate, rewriting $\langle a_i, x \rangle - b_i$ as $\langle (a_i, b_i), (x, -1) \rangle$, thereby embedding the problem in $\mathbb{R}^{n+1}$ with zero bias.

To enable effective GLM sparsification, we require that each loss function $f_i$ satisfies certain regularity conditions. Specifically, we focus on losses that are $L$-auto-Lipschitz and lower $\theta$-homogeneous---two mild and widely satisfied properties that imply several other desirable characteristics for sparsification. These conditions, along with their consequences, are discussed in detail by \citet{jambulapati2024sparsifying}. Following their results, we formalize this class of losses as $(L,\theta,c)$-proper functions, defined as follows:

\begin{definition}[Proper Loss Functions]\label{def:proper-loss}
A family of loss functions $\cF=\sets{f_1,\ldots,f_m}$ is called $(L,\theta,c)$-proper if for each $f_i:\RR \to \RR_+$, the transformed function $h_i = \sqrt{f_i}$ satisfies:
	\begin{enumerate}
	\item ($L$-auto-Lipschitz) $\abss{h_i\parens{x}-h_i \parens{x^\prime}}\leq L h_i \parens{x -x^\prime}$ for all $x, x ^\prime\in \RR$.
	\item (Lower $\theta$-homogeneous with $c$) $h_i\parens{\lambda x}\geq c \lambda ^\theta h_i \parens{x}$ for all $x \in \RR$ and $\lambda \geq 1$.
	\end{enumerate}
\end{definition}
We assume $c>0,L>0 ,\theta \in (0,4)$; in particular, $\theta=1$ already suffices for all applications in this paper. 

In this paper, we introduce a quantum algorithmic framework for GLM sparsification,  capable of efficiently constructing constructing $\epsilon$-approximate sparsifiers for a broad class of $(L,\theta,c)$-proper loss functions. Given query access to the data matrix and the loss functions, our method achieves quadratic speedups in the number of data points $m$ over the fastest classical algorithms, while preserving the same approximation guarantees. Our main result is summarized in the following theorem.

\begin{theorem}[Informal, see~\cref{thm:quantum-glm-sparsification} for the formal statement]
Let $F(x) := \sum_{i=1}^m f_i\left(\langle a_i, x \rangle \right)$ be a total loss function defined by vectors $a_1, \ldots, a_m \in \mathbb{R}^n$ (with sparsity $r\leq n$) and a $\parens{L,\theta,c}$-proper loss functions family $\cF =\sets{f_1, \ldots, f_m }$. For any $ \epsilon=O (1/r),  s_{\max}>s_{\min}\geq 0$, there exists a quantum algorithm that, with high probability\footnote{Throughout this
  paper, we say something holds ``with high probability'' if it holds
  with probability at least $1-O\paren{1/n}$.}, constructs an $\widetilde O \parens{{\log \parens{s_{\max}/s_{\min}}} \cdot n/\epsilon^2 }$-sparse $\epsilon$-approximate sparsifier of $F$ over the range $[s_{\min},s_{\max}]$, in time $\widetilde O \paren{\parens{r\sqrt{mn}/\epsilon + \poly\parens{n}}\cdot {\log \parens{s_{\max}/s_{\min}}}}$.
\end{theorem}

We illustrate applicability with common loss functions. It is straightforward to verify that $\ell_p(x)=|x|^p$ for $p\in(0,2]$ is $(L,\theta,c)$-proper with $L=1$, $\theta=p/2$, and $c=1$. For the more refined function $\gamma_p$  defined for $p \in (0, 2]$ by
\begin{equation}\label{eq:def-gamma-p}
	\gamma_p\parens{x}:=\begin{cases}
		\frac{p}{2}x^2 & \text{ for }\abss{x}\leq1,
		\\\abss{x}^p -\parens{1-p/2}& \text{otherwise},
	\end{cases}
\end{equation}
one similarly checks $(L,\theta,c)$-properness with the same parameters, where the case $p=1$ recovers the Huber loss. Regarding the dependence on the scale ratio $s_{\max}/s_{\min}$, when all $f_i$ are $p$-homogeneous the dependence disappears, and for $\gamma_p$  the  scale ratio reduces to $\poly(m)$.

\begin{table}[htb]
\caption{Comparison of runtimes for various regression problems in classical algorithms, previous quantum algorithms, and our work. Due to space constraints, we present only simplified expressions of the time complexities for each task; detailed discussions are provided in \cref{subsec:application}. We list only the state-of-the-art results that allow direct comparison and omit data-dependent algorithms.}\label{tab:regression-comparison}
\centering
\resizebox{\textwidth}{!}{%
\begin{tabular}{lccc}
\\\toprule
\textbf{Regression type} & \textbf{Classical} & \textbf{Quantum (previous)} & \textbf{Quantum (this work)} \\
\midrule
Linear regression   
& \makecell{$\widetilde{O}(mr+n^3)$ \\ \citep{nelson2013osnap,clarkson2017low}}   
& \makecell{$\widetilde O(\sqrt{m}n ^{1.5}/\epsilon +\poly(n,1/\epsilon))$ \\ \citep{song2023revisiting}}   
& \makecell{$\widetilde{O}\parens{r\sqrt{mn}/\epsilon+n^3}$ \\(\cref{cor:linear-regression})} \\
\midrule
Multiple regression   
& \makecell{$\widetilde{O}(mr+N\poly\parens{n,1/\epsilon})$ \\ \citep{nelson2013osnap,clarkson2017low}}   
& \makecell{$\widetilde O(\sqrt{m}n ^{1.5}/\epsilon +N\poly(n,1/\epsilon))$ \\ \citep{song2023revisiting}}      
& \makecell{$\widetilde{O}\parens{r\sqrt{mn}/\epsilon+N\poly\parens{n,1/\epsilon}}$\\(\cref{cor:multiple-regression}) } \\
\midrule
Ridge regression      
& \makecell{$\widetilde{O}(mr+\poly\parens{n,1/\epsilon})$ \\ \citep{avron2017sharper}}   
& \makecell{$\widetilde O(\sqrt{m}n ^{1.5}/\epsilon +\poly(n,1/\epsilon))$ \\ \citep{song2023revisiting}}    
& \makecell{$\widetilde{O}\parens{r\sqrt{mn}/\epsilon+n^3}$\\(\cref{cor:ridge-regression}) } \\
\midrule
Lasso regression   
& \makecell{$\widetilde{O}(mn^2+n^3)$ \\ \citep{efron2004least}}   
& ---   
& \makecell{$\widetilde{O}\parens{r\sqrt{mn}/
\epsilon
+\poly\parens{n,1/\epsilon}}$ \\(\cref{cor:lasso-regression})} \\
\midrule
Huber regression      
& \makecell{$\widetilde{O}(mr+\poly\parens{n,1/\epsilon})$ \\ \citep{jambulapati2024sparsifying}}    
& ---   
& \makecell{$\widetilde{O}\parens{r\sqrt{mn}/
\epsilon
+\poly\parens{n,1/\epsilon}}$\\(\cref{cor:gamma-p-regression}, with $p=1$) } \\
\midrule
$\ell_p$ regression 
& \makecell{$\widetilde{O}(mr+\poly\parens{n,1/\epsilon})$ \\ \citep{jambulapati2024sparsifying}}  
& ---   
& \makecell{$\widetilde{O}\paren{r\sqrt{mn}/
\epsilon
+\poly\parens{n,1/\epsilon}}$ \\(\cref{cor:ell-p-regression})} \\
\midrule
$\gamma_p$ regression  
& \makecell{$\widetilde{O}(mr+\poly\parens{n,1/\epsilon})$ \\ \citep{jambulapati2024sparsifying}}   
& ---   
& \makecell{$\widetilde{O}\parens{r\sqrt{mn}/
\epsilon
+\poly\parens{n,1/\epsilon}}$\\(\cref{cor:gamma-p-regression}) } \\
\bottomrule
\\
\end{tabular}
}
\end{table}

As an application of our framework, we achieve quantum speedups for a range of important regression problems, including linear regression, multiple regression, Lasso, Ridge, Huber, $\ell_p$ regression, and $\gamma_p$ regression. Since the sparsifier has size $\widetilde{O}(n/\epsilon^{2})$, it must be smaller than $m$ (otherwise it would be dense), which requires $\epsilon =\Omega\parens{ \sqrt{n/m}}$. In this regime, the leading term in our quantum runtime, $r\sqrt{mn}/\epsilon$, is strictly smaller than the classical $mr$; when $\epsilon$ is constant, this yields quadratic speedups in the parameter $m$, which often dominates in large-scale regression tasks where $m \gg n$. Moreover, for linear, multiple, and Ridge regression, our algorithm also subsumes the results of \citet{song2023revisiting}, the previous known quantum algorithm for the problem. 
We emphasize that our focus is on the \emph{computational task} of approximately minimizing the loss function $F$, rather than on the \emph{statistical learning problem} of estimating an unknown model from random samples. The latter setting was studied by \citet{chen2023quantum}, who investigated quantum algorithms for Lasso and ridge regression; in our notation, their results correspond to the case $m = O(\log n / \epsilon^{2})$.
 For clarity, \cref{tab:regression-comparison} summarizes the resulting runtimes across classical algorithms, prior quantum work, and our algorithms.

\subsection{Techniques}
\paragraph{Generalized linear model sparsification.}
To systematically investigate quantum speedups for a broad class of regression problems, we adopt the generalized linear model (GLM) sparsification framework and further extend its applicability by allowing linear combinations of loss functions. At a high level, GLM sparsification replaces a large dataset with a much smaller weighted subset (the sparsifier) such that the total loss is preserved within a $(1\pm\epsilon)$ factor for all choices of the variable $x$. This enables the original optimization problem to be approximated efficiently by solving it on the sparsifier. Recently established in the breakthrough works of \citet{jambulapati2023sparsifying,jambulapati2024sparsifying}, 
 this framework unifies and generalizes sparsification techniques for a wide range of regression objectives. Its core idea is to compute, for each loss function $f_i$, an importance score capturing its contribution to the objective, and then perform importance sampling to select a small representative set, reweighting them to form the sparsifier. In the quantum setting, our approach comprises two parts: (i) computing importance scores efficiently, and (ii) performing importance sampling based on these scores.

\paragraph{Multiscale leverage score overestimates.}
Defining suitable importance scores in the GLM setting is challenging due to its high level of abstraction: rather than specifying a concrete loss function, only general conditions on $f_i$ are given. \citet{jambulapati2024sparsifying} addressed this by analyzing the iterated covering argument for $\ell_p$ loss and generalizing $\ell_p$ Lewis weights~\citep{lewis1979ellipsoids,bourgain1989approximation,cohen2015lp}, yielding an \emph{approximate weight scheme} consisting of multiple sets of weights. These can be interpreted as importance scores at different scales, connected by specific contraction properties. Building on this, we summarized the key structural ideas into the notion of \emph{multiscale leverage score overestimates} (MLSO), which serves as the GLM analogue of importance scores.

To achieve quantum speedups, we use the results of \citet{apers2023quantum} on quantum leverage score approximation as a key foundation, enabling a quantum implementation of the contractive algorithm from \citet{jambulapati2024sparsifying}. This yields a quantum data structure, built with sublinear-time preprocessing, that supports efficient queries to MLSO values. Conceptually, our algorithm iteratively refines an initial coarse weight estimate over $\widetilde{O}(1)$ rounds, followed by a final quantum refinement step to produce MLSO values. To support this process, we develop a quantum procedure for constructing the initial weight vector by slightly modifying each $f_i$ and generating an appropriate starting weight, ensuring that the entire pipeline remains fully quantum.

\paragraph{Quantum importance sampling.}
Once quantum query access to the importance scores is available, the next step is to perform importance sampling efficiently. The main challenge is that we do not have direct access to the normalized sampling distribution, but only to an oracle returning unnormalized importance scores. Our approach, inspired by the quantum hypergraph sparsification algorithm of \citet{liu2025quantum}, uses the technique of preparing many copies of a quantum state~\citep{hamoudi2022preparing} whose amplitudes are proportional to the importance scores. We combine this with a quantum sum estimation procedure~\citep{li2019sublinear} to approximate the normalization constant, enabling correct reweighting of the sampled elements. The correctness of this method follows from a careful adaptation of the guarantees established in \citet{jambulapati2024sparsifying}.

\section{Preliminaries}

\paragraph{Notations.}
For a positive integer $n \in \mathbb{Z}_+$, we use $\sqb{n}$ to denote the set ${1,2,\ldots,n}$ and $\sqb{n}_0$ to denote the set ${0,1,\ldots,n-1}$. For vectors $u,v \in \mathbb{R}^n$, we write $\langle u,v\rangle = \sum_{i \in \sqb{n}} u_i v_i$ for the standard inner product. For a vector $w \in \mathbb{R}^n$, we use $\diag(w)$ to denote the diagonal $n \times n$ matrix whose $i$-th diagonal entry equals the $i$-th element of $w$.

\paragraph{Quantum computational model.} We assume the usual quantum computational model, which is a classical system that can (i) run quantum subroutines on $O \parens{\log m}$ qubits, (ii) can make quantum queries to the input, and (iii) has access to a quantum-read/classical write RAM (QRAM) of $\poly\parens{m}$ bits. The \emph{quantum query complexity} measures the
number of quantum queries that an algorithm makes. The \emph{quantum time complexity} measures
the number of elementary classical and quantum gates, quantum queries, and single-bit QRAM operations that algorithm uses.
 For readers less familiar with the basics of quantum computing, we refer to \cref{apdx:quantum-basic}.
 
Regarding queries to the input, we consider an input matrix $A \in \mathbb{R}^{m\times n}$ with row sparsity $r \leq n$, where the matrix elements can be queried via the oracle
$\cO_A^{\text{elem}}:\ket{i}\ket{j}\ket{0}\mapsto \ket{i}\ket{j}\ket{a_{ij}}, \forall i\in \sqb{m}, j\in \sqb{n}$. In addition, the indices of the nonzero entries of each row can be queried via the oracle $\cO_A^{\text{idx}}:\ket{i}\ket{k}\mapsto \ket{i}\ket{t_{ik}}, \forall i \in \sqb{m},k\in \sqb{r_i}$, where $t_{ik}$ is the column index of the $k$-th nonzero entry in the $i$-th row of $A$, and $r_i \leq r$ is the number of nonzero entries in that row. For simplicity, we use the notation $\cO_A$ to refer to either of these two oracles. We will also use the following quantum queries: for a vector $w \in \mathbb{R}^n$, its elements can be queried via the oracle $\cO_w:\ket{i}\ket{0}\mapsto\ket{i}\ket{w_i},  i\in \sqb{n}$; and for a family of functions $\cF=\sets{f_1,\ldots,f_m}$, we assume access to the oracle $\cO_\cF:\ket{i}\ket{x}\ket{0}\mapsto\ket{i}\ket{x}\ket{f_i(x)}$.

\paragraph{Quantum leverage score approximation.}
Let $A \in \mathbb{R}^{m \times n}$ be a matrix with rows $a_1, \dots, a_m$. The \emph{leverage score} of the $i$-th row is defined as
\begin{equation}\label{eq:leverage-score}
	\sigma_i\parens{A}: = a _i ^{\top} \parens{A ^\top  A}^{+}a _i, \quad \forall i\in \sqb{m},
\end{equation}
where $(A^{\top} A)^+$ denotes the Moore-Penrose pseudoinverse of $A^{\top}A$. Intuitively, $\sigma_i(A)$ measures the ``statistical importance'' of row $i$ with respect to the whole set of rows.
Leverage scores are central to row-sampling methods that construct a smaller matrix $\widetilde{A}$ preserving the key spectral properties of $A$. Sampling rows proportionally to $\sigma_i(A)$ and rescaling yields $\widetilde{A}^\top \widetilde{A} \approx A^\top A$ with high probability, enabling accurate and efficient approximation of linear regression and other optimization tasks; we refer the reader to the survey of \citet{woodruff2014sketching} for a comprehensive overview.

In a recent important work, \citet{apers2023quantum} introduced a simple recursive halving algorithm for quantum spectral approximation. As a corollary, they obtained a quantum algorithm for leverage score approximation, which serves as a cornerstone of our work---we invoke this result multiple times throughout our algorithms. For completeness, we restate their result below.

\begin{theorem}[Quantum Leverage Score Approximation] {\cite[Theorem 3.2]{apers2023quantum}}\label{thm:quantum-leverage-score-aapproximation}
	Assume query access to a matrix $A =\RR^{m \times n}$ with row sparsity $r$. For any $0<\epsilon\leq 1$, there exists a quantum algorithm that, with high probability,  provides quantum query access to estimates $\widetilde \sigma_i $ for any $i\in \sqb{m}$ satisfying $\parens{1-\epsilon}\sigma_i\parens{A} \leq  \widetilde \sigma_i \leq \parens{1+\epsilon}\sigma_i\parens{A}$. The algorithm makes $\widetilde O \parens{\sqrt{mn}/\epsilon}$ row queries to $A$ and runs in $\widetilde O \parens{r\sqrt{mn}/\epsilon+ n^\omega+\min \sets{n ^\omega, nr ^2}/\epsilon^2+n ^2/\epsilon^4}$ time. Furthermore, the cost to estimate per $\widetilde{\sigma}_i$ is one row query to $A$ and $\widetilde{O}(r / \epsilon^2)$ time.
	\end{theorem}

We note that a row query to the $i$-th row of $A$ returns the entire row, and the query can be made in superposition.  Equivalently, one row query corresponds to $O \parens{r}$  queries to $\cO_A$, since each row contains at most $r$ nonzero entries.

\section{Quantum Algorithm for Multiscale Leverage Score Overestimates}
In this section, we present our core quantum algorithm, \emph{Quantum Multiscale Leverage Score Overestimates} (QMLSO), the key to enabling quantum acceleration in our GLM sparsification framework. We begin by defining MLSO, distilling the structural ideas from the approximate weight scheme of \citet{jambulapati2024sparsifying}, which serves as the GLM analogue of importance scores. We then describe our quantum procedure for computing MLSO values from an initial coarse weight estimate, using quantum leverage score approximation to achieve sublinear-time refinement. Finally, we develop a quantum algorithm for constructing suitable initial weights, ensuring that the entire MLSO computation can be executed fully within the quantum setting.

\subsection{Multiscale leverage score overestimates (MLSO)}
We begin by recalling the definition of approximate weights introduced by~\citet{jambulapati2024sparsifying}, which generalizes the notion of $\ell_p$ Lewis weights to capture the “statistical contribution’’ of each row at a given scale, where the scale refers to the magnitude range of residuals being considered.

\begin{definition}[Approximate Weights]Fix a matrix $A\in \RR^{m\times n}$ (where the $i$-th row is $a_i$ for each $i\in \sqb{m}$) and a family of loss functions $\cF=\sets{f_1,\ldots,f_m}$. We say a vector $w \in \RR^m_+$ is an $\alpha$-approximate weight (with respect to $A$ and $\cF$) at scale $s$ if, for all $i \in \sqb{m}$,
\begin{equation*}
	\frac{s}{\alpha}\leq \frac{f_i\parens{\norms{M^{-1/2}a_i}_2}}{w_i\norms{M^{-1/2}a_i}_2^2} \leq \alpha s, \quad \text{ where}\quad M :=\sum_{i =1}^m w_i a_i a _i ^{\top}.
\end{equation*}
\end{definition}
Since approximate weights are defined only at a single scale, they are not sufficient for handling general loss functions where homogeneity fails. To address this, \citet{jambulapati2024sparsifying} introduced weight schemes---collections of approximate weights defined consistently across a range of scales. These schemes guarantee that the contribution of each row is well controlled not only at a single scale but simultaneously across multiple scales, thereby capturing the multiscale structure inherent in $A$ and $\cF$.

\begin{definition}[Weight Schemes]Fix a matrix $A\in \RR^{m \times n}$ (the $i$-th row is $a_i$ for each $i\in \sqb{m}$) and a family of  loss functions $\cF=\sets{f_1,\ldots,f_m}$. Let $\cJ\subseteq \ZZ$ be a contiguous interval. We say that a family $\cW_{A,\cF,\cJ}=\sets{w ^{(j)}\in \RR^m_+: j \in \cJ}$ is an $\alpha$-approximate weight scheme if:
\begin{enumerate}
	\item  Each $w^{(j)}$ is an $\alpha$-approximate weight (with respect to $A$ and $f$) at scale $2^j$; 
	\item For every pair $j ,j +1\in \cJ$, we have $w ^{(j +1)}_{i}\leq \alpha w _i ^{(j)}, \forall i \in \sqb{m}$.
\end{enumerate}
We omit the subscript $(A,\cF,\cJ)$ when it is clear from context.
\end{definition}

When working with weight schemes, we ultimately need a compact summary that upper bounds the statistical contribution of each row across all relevant scales. This motivates the notion of multiscale leverage score overestimates. By aggregating leverage scores over the entire weight scheme, such overestimates provide a single vector that controls row contributions across all scales. This will later serve as the key object enabling efficient sampling and sparsification in our quantum algorithm.

\begin{definition}[Multiscale Leverage Score Overestimates]\label{def:mlso}
Let $A \in \RR^{m \times n}$ be a matrix with rows $a_1, \ldots, a_m$, let $\cF=\{f_1,\ldots,f_m\}$ be a family of loss functions, and let $\cJ \subseteq \ZZ$ be a contiguous interval. Suppose $\cW=\sets{w^{(j)}\in \RR^m_+: j \in \cJ}$ is an $\alpha$-approximate weight scheme  with respect to $A,\cF,\cJ$. We say that a vector $z\in \RR^m_+$ is a multiscale (leverage score) $\tau$-overestimates with respect to $A$ and $\cW$ if 
\begin{equation*}
	\norms{z}_1\leq \tau,\quad\text{ and }\quad  z_i\geq \sigma_i^{\multiscale}\parens{A,\cW},\quad\forall i\in \sqb{m},
\end{equation*}
where the multiscale leverage score of the $i$-th row is defined as
\begin{equation*}
	\sigma_i^{\multiscale}\parens{A,\cW} := \max_{j \in \cJ} \sigma_i\parens{W_{j}^{1/2} A}, \quad \forall i\in \sqb{m},
\end{equation*}
and $W_j := \diag (w^{(j)})$ for each $j \in \cJ$.
\end{definition}

\subsection{Quantum algorithm for MLSO}
We now present our quantum algorithm for computing multiscale leverage score overestimates (MLSO), which serves as a key subroutine in our main quantum algorithm for GLM sparsification.  
Our approach builds upon the classical contractive algorithm of~\citet{cohen2015lp,jambulapati2024sparsifying}, but is carefully adapted to the quantum setting.  

The input consists of quantum oracles for the matrix $A$ and the function family $\cF$, along with query access to an initial approximate weight vector $w^\circ$ and relevant parameters.  
The output is a data structure that provides efficient quantum query access to multiscale leverage score overestimates (see~\cref{prop:preparation-multiscale-overestimates} for a formal description).  

The algorithm proceeds iteratively: over $\widetilde O(1)$ rounds, it refines the largest-scale weight until it becomes well-conditioned. Once this condition is achieved, the algorithm recursively computes the weights at smaller scales and combines them to obtain the final MLSO.  
Within each iteration, the subroutine $\mathsf{ModLevApprox}$ supplies quantum query access to the modified leverage scores (\cref{prop:modlevapprox}), while $\mathsf{WeightCompute}$ produces quantum query access to the updated weights for the subsequent round (\cref{prop:weight-computation}).  
The complete procedure is summarized in \cref{alg:qmlso}, and its formal guarantees are established in \cref{thm:qmlso}.
Owing to space limitations, the proof of \cref{thm:qmlso} is deferred to \cref{apdx:qmlso}.

\begin{algorithm}[htb]
  \caption{$\textsc{QMLSO}(\cO_\cF,\cO_A, \cO_{w^\circ}, j_{\min},j _{\max}, \beta, L,\theta, c)$}\label{alg:qmlso}
  \begin{algorithmic}[1] 
    \REQUIRE{}  A matrix oracle $\cO_A$; a proper loss functions oracle $\cO_\cF$ with parameters ${L,\theta,c}$; an initial weight oracle $\cO_{w^\circ}$  with associated parameters $\beta$ and $j_{\max}$; and an  integer $j_{\min}<j_{\max}$.
    \ENSURE{} An instance $\cZ$ of $\mathsf{QOverestimate}$ which stores the vector $z$ being a multiscale $O \parens{n\cdot\abss{j_{\max}-j_{\min}}}$-overestimate with respect to matrix $A$ and $\alpha$-approximate scheme $\cW=\sets{w ^{(j)}\in \RR^m_+ : j  \in \cJ:= \ZZ\cap [j_{\min},j_{\max}]}$ for some $\alpha=\alpha\paren{L,\theta,c}$.
    \STATE{} $\epsilon \gets 0.1$, $\delta\gets \max\set{\frac{1}{2}, \abss{\frac{\theta-2}{2}}}$, $C\gets\max\set{\frac{2L}{c} ,\frac{1}{c}}$
 \STATE{} $T\gets \ceil{\paren{\log\log\parens{\frac{1-\epsilon}{\parens{1+\epsilon}^2}\cdot C} -\log\log \beta} /\log \delta}$
   	\STATE{} $U_{w_{\text{itr}}^{(0)}}\gets \cO_{w^\circ}$
    \FOR{$i=0$ to $T$} 
    \STATE{} $\cL^{(i)}_{\text{itr}}\gets \mathsf{ModLevApprox}\parens{\cO_{A},U_{w^{(i)}_{\text{itr}}},\epsilon}$
    \STATE{} $U_{w^{(i+1)}_{\text{itr}}}\gets \mathsf{WeightCompute}\parens{\cO_\cF,U_{w^{(i)}_{\text{itr}}},\cL^{(i)}_{\text{itr}},2^{j_{\max}}}$
    \ENDFOR{}
    \STATE{} $U_{w^{(j_{\max})}} \gets U _{w_{\text{itr}}^{(T)}}$
	\FOR{$j=j_{\max}-1,\ldots $ to $j_{\min}$} 
    \STATE{} $\cL^{(j+1)}  \gets \mathsf{ModLevApprox}\parens{\cO_{A},U_{w^{(j+1)}},\epsilon}$
    \STATE{} $U_{w^{(j)}}\gets \mathsf{WeightCompute}\parens{\cO_\cF,U _{w^{(j+1)}}, {\cL}^{(j+1)} ,2^{j}}$
    \ENDFOR{}
	\STATE{} $\cZ\gets\mathsf{QOverestimate}\parens{\sets{ \cL ^{(j)}: j \in \cJ},\epsilon}$
  \end{algorithmic}
\end{algorithm}

\begin{theorem}[Quantum Multiscale Leverage Score Overestimates]\label{thm:qmlso}
There is a quantum algorithm $\textsc{QMLSO}(\cO_\cF, \cO_A, \cO_{w^\circ},j_{\min}, j_{\max},\beta,L,\theta,c)$ that, given query access  $\cO_A$ to a matrix $A=\RR^{m \times n}$ with row sparsity $r\leq n$, query access $\cO_\cF$ to a $\parens{L,\theta,c}$-proper loss functions family $\cF =\sets{f_1, \ldots, f_m }$, query access $\cO_{w^\circ}$ to a initial weight $w^\circ\in \RR^m_+$ which is a $\beta$-approximate weight at scale $2^{j_{\max}}$ (with respect to $A$ and $\cF$) for some integer $j_{\max}$, and any integer $j_{\min}<j_{\max}$, the algorithm makes $O\parens{\sqrt{mn}}$ queries to $\cO_{w^\circ}$, $O \parens{r\sqrt{mn}\Delta}$ queries to $\cO_{A}$, $O \parens{\sqrt{mn}\Delta}$ queries to $\cO_{\cF}$, runs in time $\widetilde O \paren{\parens{r\sqrt{mn}+n ^\omega+nr^2 }\Delta}$, where $\Delta=\abs{j_{\max}-j_{\min}}+\max\sets{\log\log\beta,0}$. Then, with high probability, it provides query access to a multiscale $O\paren{n\cdot\abss{j_{\max}-j_{\min}} }$-overestimates $z$ with respect to $A$ and  $\alpha$-approximate weight scheme  $\cW=\sets{w ^{(j)}\in \RR^m_+ : j \in \cJ:= \ZZ\cap [j_{\min},j_{\max}]}$ for some $\alpha=\alpha\parens{L,\theta,c}$, where each query to $z$ requires $\widetilde O\paren{r\abss{\cJ}}$ time.
\end{theorem}

\subsection{Initial Weight Construction}

When applying our QMLSO algorithm, it is crucial to start with a well-conditioned $\beta$-approximate weight, where $\beta$ cannot be too large.  
In this subsection, we show how to construct such an initial weight by leveraging quantum leverage score approximation together with a binary search procedure, combined with a mild modification of the loss functions $f_i$.  
Our initialization algorithm outputs a quantum data structure that, given query access to $A$ and $\mathcal{F}$, as well as the largest scale $s_{\max}$ and a parameter $\delta$, provides quantum query access to an initial weight vector $w^{\circ}$ (at scale $s_{\max}$) together with the corresponding modified loss family $\mathcal{F}^{\circ}$.  
This construction can be viewed as the quantum analogue of the initialization procedure introduced by~\citet{jambulapati2024sparsifying}.
We formalize this construction in the following theorem, and defer its proof to \cref{apdx:weight-initialize}.

\begin{theorem}[Weight Initialization]\label{thm:weight-initialize}
Assume query access $\cO_{A}$ to a matrix $A =\RR^{m \times n}$  with row sparsity $r$ and query access $\cO_\cF$ to a $\parens{L,\theta,c}$-proper loss functions family $\cF =\sets{f_1, \ldots, f_m }$. There exists a quantum data structure $\mathsf{WeightInitialize}$ that supports the  the following operations:
\begin{enumerate}
	\item Initialization: $\mathsf{WeightInitialize}\parens{\cO_{A},\cO_{\cF},s_{\max},\delta}$, outputs an instance $\cA$, 
	 making $\widetilde O \parens{r\sqrt{mn}}$  queries to $\cO_A$ and $\widetilde O \parens{ \sqrt{mn}}$ queries to $\cO_\cF$,  in $\widetilde O \parens{r\sqrt{mn} +n^\omega+nr^2 }$ time.
	\item Function query: $\cA.\mathsf{QueryFunction}$, provides a query access $\cO_{\cF^\circ}$ to a $\parens{\max\sets{1,L},\theta,c}$-proper loss functions family $\cF^\circ=\sets{f_1^\circ,\ldots,f_m^\circ}$ which satisfies
\begin{equation}\label{eq:weight-initialize}
	\sum_{i=1}^m f_i \paren{\ips{a_i}{x}}\leq s_{\max}\quad\Rightarrow \quad 0\leq \sum_{i=1}^m f_i^\circ \paren{\ips{a_i}{x}}-f_i \paren{\ips{a_i}{x}}\leq C_{\textsc{init}} \delta m^2 s_{\max},
\end{equation}
for some constant $C_{\textsc{init}}=C_{\textsc{init}}\parens{L,\theta,c}$. Each query requires $O (r)$ queries to $\cO_A$ and $O \parens{1}$ queries to $\cO_\cF$, and runs in $\widetilde{O}\parens{r}$ time.

	\item Weight query: $\cA.\mathsf{QueryWeight}$, provides a query access $\cO_{w^{\circ}}$ to a weight $w^\circ\in \RR^m_+$ which is a $\beta$-approximate weight at scale $s_{\max}$ with respect to $A$ and $\cF^\circ$, where $\beta= O \parens{(L/c)^2 m/\delta}$.
	Each query requires $O (r)$ queries to $\cO_A$ and $O (1)$ queries to $\cO_\cF$, and runs in $\widetilde{O}\parens{r}$ time.
\end{enumerate}
\end{theorem}

\section{Quantum Algorithm  for Generalized Linear Models Sparsification}
In this section, we present our \emph{Quantum Algorithm for GLM Sparsification}, which builds on the MLSO procedure to achieve substantial quantum speedups. The first part outlines our quantum sparsification framework,  which implements efficient importance sampling by combining the preparation of many copies of a quantum state~\citep{hamoudi2022preparing} with quantum sum estimation~\citep{li2019sublinear}, enabling sublinear-time sampling and reweighting. 
 The second part demonstrates the versatility of this framework by embedding a broad class of optimization problems into it, showing that they can be solved with quantum speedups.

\subsection{Quantum Algorithm for GLM Sparsification}

To describe our quantum sparsification framework, we rely on two fundamental primitives. The first is Multiple Quantum State Preparation, which efficiently produces many copies of a quantum state; measuring these copies yields samples where each index $i$ is selected with probability proportional to $w_i$, thereby enabling importance sampling with sublinear overhead (restated in \cref{thm:quantum-prob-sample}). The second is Quantum Sum Estimation, which offers a quadratic speedup over classical approaches and is crucial for efficient reweighting in our framework (restated in \cref{thm:quantum-sum-estimate}).

With these tools in place, we now present our main result: the Quantum Algorithm for GLM Sparsification. Building on the QMLSO procedure and the primitives above, our algorithm efficiently sparsifies generalized linear models while achieving polynomial quantum speedups. The algorithm and its formal guarantee are stated below, with the proof deferred to \cref{apdx:glm-sparsification}.

\begin{algorithm}[htb]
	\caption{Quantum GLM Sparsification $\mathsf{QGLMSparsify}\parens{\cO_{\cF},\cO_A,L,\theta,c,\epsilon,s_{\min },s_{\max}}$}\label{alg:hyper-sparse}
	\begin{algorithmic}[1]
    	\REQUIRE{} A matrix oracle $\cO_A$; a proper loss functions oracle $\cO_\cF$ with parameters ${L,\theta,c}$; parameters $\epsilon>0,s_{\max}>s_{\min}\geq 0$.
		\ENSURE{} A non-negative weight vector $w\in \RR^m_+$ such that $\widetilde F(x):=\sum_{i=1}^m w_i f_i\parens{\ips{a_i}{x}}$ is an $\widetilde O \parens{{\log \parens{s_{\max}/s_{\min}}}\cdot  n/\epsilon^2}$-sparse $\epsilon$-approximate sparsifier of $F$ over the range $\sqbs{s_{\min},s_{\max}}$.
		\STATE{} $ \widetilde w =0, t \gets 1,  M \gets \widetilde\Theta \parens{n   /\epsilon^2}, \delta\gets O \parens{\epsilon s_{\min}/\parens{m^3 s_{\max}}}$
		\STATE{} $j_{\min}\gets \floors{\log s_{\min}-4\log m } ,j _{\max}\gets \ceils{\log s_{\max}}, \beta\gets O \parens{(L/c)^2m/\delta}$
		\STATE{}  $\cA \gets \mathsf{WeightInitialize}\parens{\cO_{A},\cO_{\cF},s_{\max},\delta}$
		\hfill (\cref{thm:weight-initialize})
		\STATE{}  $\cZ \gets \mathsf{QMLSO}(\cA.\mathsf{QueryFunction},\cA.\mathsf{QueryWeight} ,j_{\min} ,j_{\max},\beta,L,\theta,c)$
		\hfill (\cref{thm:qmlso})
		\STATE{} $\vartheta\gets \mathsf{MultiSample}(\cZ, M)$
		\hfill (\cref{thm:quantum-prob-sample})
		\STATE{} $\widetilde \nu \gets \mathsf{SumEstimate}(\cZ, 0.1)$
		\hfill (\cref{thm:quantum-sum-estimate})
		\FOR{$i=1$ to $M$}
		\STATE{} Get $z_{\vartheta_i}$ from $\cZ$
		\STATE{} $\widetilde w _{\vartheta_i} \leftarrow \widetilde w _{\vartheta_i}+ \widetilde \nu/\parens{1.1 \, M z_{\vartheta_i}  } $
		\ENDFOR{}
		\STATE{} Return $\widetilde w$.
	\end{algorithmic}
\end{algorithm}

\begin{theorem}[Quantum Algorithm for GLM Sparsification]\label{thm:quantum-glm-sparsification}
Let $F(x) := \sum_{i=1}^m f_i\left(\langle a_i, x \rangle \right)$ be a total loss function defined by vectors $a_1, \ldots, a_m \in \mathbb{R}^n$ (with sparsity $r\leq n$) and a $\parens{L,\theta,c}$-proper loss functions family $\cF =\sets{f_1, \ldots, f_m }$. There exists a quantum algorithm
	$\mathsf{QGLMSparsify}\parens{\cO_A,\cO_{\cF},L,\theta,c,\epsilon,s_{\min},s_{\max}}$ that, given query access $\cO_A$ to the matrix $A$ (where the $i$-row is $a_i$),  query access $\cO_\cF$ to the function family $\cF$ with parameters $L,\theta, c$, and numbers $\epsilon>0$, $s_{\max}>s_{\min}\geq 0$, outputs with high probability a non-negative weight vector $w\in \RR^m_+$ such that $\widetilde F:=\sum_{i=1}^m w_i f_i\parens{\ips{a_i}{x}}$ is an $\widetilde O \parens{\log \parens{s_{\max}/s_{\min}} \cdot n/\epsilon^2}$-sparse $\epsilon$-approximate sparsifier of $F$ over the range $[s_{\min},s_{\max}]$, in time $\widetilde O \paren{\parens{n ^\omega+nr^2 +r\sqrt{mn}/\epsilon }\cdot {\log \parens{s_{\max}/s_{\min}}}}$.
\end{theorem}
We note that  scale ratio term $\log\parens{s_{\max}/s_{\min}}$ can be removed when all $f_i$ are $p$-homogeneous for some $p>0$ (i.e., $f_i(\lambda x)=\abs{\lambda}^p f_i(x)$, and in particular $\ell_p$ is $p$-homogeneous). For non-homogeneous cases such as $\gamma_p$, this  can  be reduces to $\log(\poly(m)) $ via properness assumption ~\citep[Lemma 1.5]{jambulapati2024sparsifying}.

\subsection{Applications}\label{subsec:application}
Firstly, it's not hard to verify that both $\ell_p(x)=\abss{x}^p$ and $\gamma_p$ (as defined in \cref{eq:def-gamma-p}) are $(L,\theta,c)$-proper for $p\in(0,2]$, with parameters $L=1$, $\theta=\tfrac{p}{2}$, and $c=1$. Applying \cref{thm:quantum-glm-sparsification}, this yields a $\widetilde{O}(n/\epsilon^2)$-sparse $\epsilon$-approximate  sparsifier for the $\ell_p$ and $\gamma_p$ losses, computable in time $\widetilde{O}\left(r\sqrt{mn}/\epsilon+n^\omega+nr^2\right)$. Given such an $\epsilon$-approximate sparsifier, one can then invoke a fast classical algorithm to obtain an $\epsilon$-accurate solution in $\poly(n,1/\epsilon)$ time. Notably, the case $p=2$ for $\ell_p$ recovers linear regression (\cref{cor:linear-regression}), while $p=1$ for $\gamma_p$ corresponds to Huber regression.

\begin{corollary}[Quantum $\ell_p$ Regression]\label{cor:ell-p-regression}
Let $p\in (0,2]$. There exists a quantum algorithm that, given query access to a matrix $A\in \RR^{m\times n}$ (with row sparsity $r \leq n$) and vector $b\in \RR^m$, and $\epsilon> 0$, outputs with high probability an $x\in \RR^n$ such that $\norms{Ax-b}^p_p \leq \parens{1+\epsilon} \min _{x\in\RR^n }\norms{Ax-b}^p_p$, in $\widetilde O \parens{r\sqrt{mn}/\epsilon}+\poly \parens{n,\epsilon}$ time. 
\end{corollary}
\begin{corollary}[Quantum $\gamma_p$ Regression]\label{cor:gamma-p-regression}
Let $p\in (0,2]$. There exists a quantum algorithm that, given query access to a matrix $A\in \RR^{m\times n}$ where the $i$-th row is $a_i$ (with sparsity $r \leq n$) and vector $b\in \RR^m$, and $\epsilon> 0$, outputs with high probability an $x\in \RR^n$ such that $F\parens{x} \leq \parens{1+\epsilon} \min _{x\in\RR^n } F\parens{x}$ for $F \parens{x}=\sum_{i=1}^m \gamma_p \parens{\ips{a_i}{x}-b_i}$, in $\widetilde O \parens{r\sqrt{mn}/\epsilon}+\poly \parens{n,\epsilon}$ time. 
\end{corollary}

Now we go further and consider the case of multiple regression. Given $A\in \RR^{m\times n}$ and $B\in \RR^{m\times N}$, the goal is to find $X\in \RR^{n \times N}$ such that
$\norm{AX-B}_F \leq (1+\epsilon)\norm{AX^\prime-B}_F$.
Writing $X=\sqb{x^1,\ldots,x^N}$ and $B=\sqb{b^1,\ldots,b^N}$, we expand
$
	\norms{AX-B}_F^2=\sum_{k=1}^{N}\norms{ Ax^k-b^k}^2=\sum_{k=1}^N\sum_{j=1}^{m}\parens{\ips{a_i}{x^k}-B_{jk}}^2.
$
Thus, multiple regression reduces to a single large-scale linear regression problem, with an effective data dimension $n^\prime= nN$. See \cref{cor:multiple-regression} for the formal description of the quantum algorithm.

We next turn to the ridge regression objective
$
F_{\mathrm{ridge}}(x) =\|Ax-b\|_2^2 + \lambda \|x\|_2^2.
$
The regularization term can be written as a sum of coordinate-wise quadratic losses:
$
\lambda \|x\|_2^2= \sum_{j=1}^n \big(\sqrt{\lambda}\langle e_j, x \rangle\big)^2.
$
For using the GLM sparsification framework, we define
\begin{equation}
\label{eq:ridge}
A^\prime=\begin{bmatrix}A\\ \sqrt{\lambda}I_n\end{bmatrix}\in\mathbb{R}^{(m+n)\times n}, 
\qquad
b^\prime=\begin{bmatrix}b\\ 0_n\end{bmatrix}\in\mathbb{R}^{m+n},
\end{equation}
and use the  quadratic loss $f(x)=x^2$. Then for every $x\in\mathbb{R}^n$ we have
$
\sum_{i=1}^{m+n} f\big(\langle a^\prime_i,x\rangle - b^\prime_i\big)
= \|Ax-b\|_2^2 + \lambda\|x\|_2^2
= F_{\mathrm{ridge}}(x).
$ Hence, ridge regression reduces to linear regression with an effective data size $m^\prime = m+n$. The formal quantum algorithm is given in \cref{cor:ridge-regression}.

Beyond ridge, the lasso regression objective introduces $\ell_1$ regularization: $F_{\mathrm{lasso}}(x)=|Ax-b|_2^2+\lambda|x|_1$. We define the loss functions
$\cF^\prime=\sets{f^\prime_1\parens{x}=\cdots=f^\prime_m\parens{x} ={x}^2; f^\prime_{m+1}\parens{x}=\cdots=f^\prime_{m+n}\parens{x}=\abss{x}_1}$, and adopt the notation from \cref{eq:ridge}.
It follows that
$\sum_{i=1}^{m+n} f^\prime \parens{\ips{a^\prime_i}{x}-b^\prime_i} =\norms{Ax-b}^2_2+\lambda \norms{x}_1$.
Hence, the GLM sparsification framework applies. The corresponding quantum algorithm is formally described in \cref{cor:lasso-regression}.

\bibliographystyle{plainnat}
\bibliography{ref}
\appendix

\section{Quantum computing}\label{apdx:quantum-basic}
In quantum mechanics, a $d$-dimensional quantum state
$\ket{v} = (v_0,\ldots,v_{d-1})^\top$ is represented as a unit vector in the Hilbert space $\mathbb{C}^d$, satisfying $\sum_{i\in\sqb{d}_0} |v_i|^2=1$.
The computational basis of this space is $\set{\ket{i}}_{i\in\sqb{d}_0}$, where
$\ket{i}=(0,\ldots,0,1,0,\ldots,0)^\top$ with the $i$-th coordinate (0-indexed) equal to 1 and all others 0.
For two states $\ket{u},\ket{v}\in\mathbb{C}^d$, their inner product is
$\langle u|v\rangle = \sum_{i\in\sqb{d}_0} u_i^* v_i$, with $z^*$ denoting complex conjugation.
The tensor product of $\ket{u}\in\mathbb{C}^{d_1}$ and $\ket{v}\in\mathbb{C}^{d_2}$ is their Kronecker product,
$\ket{u}\otimes \ket{v} = (u_0 v_0, u_0 v_1,\ldots,u_{d_1-1}v_{d_2-1})^\top$, 
often abbreviated as $\ket{u}\ket{v}$.

A quantum bit, or qubit, is simply a  quantum state in $\mathbb{C}^2$, expressible as $\ket{\psi}=\alpha\ket{0}+\beta\ket{1}$, where $\alpha,\beta\in\mathbb{C}$ and $|\alpha|^2+|\beta|^2=1$.
An $n$-qubit state belongs to the tensor product space $(\mathbb{C}^2)^{\otimes n}=\mathbb{C}^{2^n}$, with basis states ${\ket{i}}_{i\in\sqb{2^n}_0}$.
When measuring an $n$-qubit state $\ket{\psi}$ in this basis, the outcome $i$ is observed with probability $|\langle i|\psi\rangle|^2$.
Quantum operations are described by unitary operators $U$, which satisfy $UU^\dagger=U^\dagger U=I$, where $U^\dagger$ denotes the conjugate transpose and $I$ is the identity operator.

In many quantum algorithms, information is stored and accessed through quantum-read classical-write random access memory (QRAM)~\citep{giovannetti2008quantum}, which is employed in numerous quantum algorithms. QRAM allows the storage and modification of a classical array $a_1, \ldots, a_n$ while supporting quantum query access via the unitary
$U_{\textup{QRAM}}:\ket{i}\ket{0}\mapsto \ket{i}\ket{a_i}$.
If one is only concerned with quantum query complexity, the QRAM assumption can be removed at the cost of a polynomial increase in time complexity. While QRAM serves as a natural quantum analogue of the classical RAM model and is frequently employed in theoretical work, it is important to note that the practical realization of scalable QRAM remains highly uncertain given the current state of quantum hardware.

A quantum (query) algorithm $\cA$ is a quantum circuit consisting of a sequence
of unitaries $U_1, \ldots, U_T$, where each $U_t$ may represent a quantum gate, a
quantum oracle, or a QRAM operation.
The time complexity of $\mathcal{A}$ is measured by the total number $T$ of such operations.
The algorithm acts on $n$ qubits, initialized in the state $\ket{0}^{\otimes n}$.
Applying the unitaries sequentially produces the final state
$\ket{\psi} = U_T \dots U_1 \ket{0}^{\otimes n}$.
Finally, a measurement in the computational basis ${\ket{i}}_{i\in\sqb{2^n}_0}$ then yields a classical outcome $i$ with probability $|\langle i | \psi \rangle|^2$.

\section{Proof of~\cref{thm:qmlso}}\label{apdx:qmlso}

\subsection{Quantum Subroutines}

To present our algorithm more clearly, we introduce the following data structure, which can be viewed as a straightforward application of \cref{thm:quantum-leverage-score-aapproximation}.

\begin{proposition}[Quantum Modified Leverage Score Approximation]\label{prop:modlevapprox}
Assume query access $\cO_{A}$ to a matrix $A =\RR^{m \times n}$ with row sparsity $r$ and query access $\cO_w$ to a vector $w \in \RR^{m}_+$. For any $0<\epsilon\leq 1$, there exists a quantum data structure $\mathsf{ModLevApprox}$ that supports the  the following operations:
\begin{enumerate}
	\item Initialization: $\mathsf{ModLevApprox}\parens{\cO_{A},\cO_{w},\epsilon}$, outputs an instance $\cL$,  making $\widetilde O \parens{r\sqrt{mn}/\epsilon}$  queries to $\cO_A$ and $\widetilde O \parens{\sqrt{mn}/\epsilon}$ queries to $\cO_w$, in $\widetilde O \parens{r\sqrt{mn} /\epsilon+ \parens{n^\omega+nr^2}/\epsilon^2+ n^2/\epsilon^4 }$ time.
	\item Query: $\cL.\mathsf{Query}$, outputs an unitary satisfying
		\begin{equation*}
			\cL.\mathsf{Query}\ket{i}\ket{0}=\ket{i}\ket{\widetilde \sigma_i},
		\end{equation*}
		where for all $i \in \sqb{m}$, it holds that $(1-\epsilon)\sigma_i \leq \widetilde{\sigma}_i \leq (1+\epsilon)\sigma_i$, with $\sigma_i = \sigma_i\parens{W^{1/2}A}$ denoting the leverage score, and $W = \diag\parens{w}$. Each query requires $O \parens{r}$ queries to $\cO_A$ and $O(1)$ queries to $\cO_w$, and runs in $\widetilde{O}\parens{r/\epsilon^2}$ time.
\end{enumerate}
\end{proposition}

\begin{proof}[Proof of \cref{prop:modlevapprox}]
Write \(B := W^{1/2}A\in\RR^{m\times n}\), where
\(W = \diag(w)\).
Because \(W^{1/2}\) rescales each row of \(A\) by \(\sqrt{w_i}\),
\(B\) has the \emph{same} row sparsity \(r\). Let $A_i$ be the $i$-th row of $A$.
Fix a row index \(i\) and let \(S_i=\operatorname{supp}(A_{i\cdot})\) with
\(|S_i|=r_i\leq r\).
All entries of $A_i$ can be obtained using $O (r)$ queries to $\cO_A$, one query to $\cO_w$, and $\widetilde O(r)$ arithmetic time. For the superposition, we can perform
\[
\begin{aligned}
\ket{i}\ket{0}&\xmapsto{\cO_A^{\text{idx}}}\;
\frac{1}{\sqrt{r_i}}\sum_{j\in S_i}\ket{i}\ket{j}\ket{0}\\&\xmapsto{\cO_A^{\text{elem}}}\;
  \frac{1}{\sqrt{r_i}}\sum_{j\in S_i}\ket{i}\ket{j}\ket{A_{ij}}\ket{0}\\
&\xmapsto{\cO_w}\;
  \frac{1}{\sqrt{r_i}}\sum_{j\in S_i}\ket{i}\ket{j}\ket{A_{ij}}\ket{w_i}\ket{0}\\
&\xmapsto{U_\textsc{mult}}\;
  \frac{1}{\sqrt{r_i}}\sum_{j\in S_i}\ket{i}\ket{j}\ket{A_{ij}}\ket{w_i}\ket{\sqrt{w_i}A_{ij}}\\
&\xmapsto{\cO_w^{\dagger},{\cO_A^{\text{elem}}}^{\dagger},{\cO_A^{\text{idx}}}^{\dagger}}\;
  \frac{1}{\sqrt{r_i}}\sum_{j\in S_i}\ket{i}\ket{j}\ket{0}\ket{0}\ket{B_{ij}} .
\end{aligned}
\]
Thus one row query in \cref{thm:quantum-leverage-score-aapproximation}
corresponds \(O (r)\)  queries to \(\cO_A\), \(O (1)\) queries to \(\cO_w\), plus \(\widetilde O(r)\) arithmetic time.
Therefore, invoking the \cref{thm:quantum-leverage-score-aapproximation} on \(B\) with accuracy \(\epsilon\) gives, with high
probability:
(i) an initialisation that performs \(\widetilde O(\sqrt{mn}/\epsilon)\) row
queries to \(B\) and runs in
\(\widetilde O\bigl(r\sqrt{mn}/\epsilon+(n^{\omega}+nr^{2})/\epsilon^{2}+n^{2}/\epsilon^{4}\bigr)\);
substituting the row-query cost yields
\(\widetilde O(r\sqrt{mn}/\epsilon)\) calls to \(\cO_A\) and
\(\widetilde O(\sqrt{mn}/\epsilon)\) calls to \(\cO_w\);
(ii) a query unitary that maps \(\ket{i}\ket{0}\) to
\(\ket{i}\ket{\widetilde\sigma_i(B)}\) with
\((1-\epsilon)\sigma_i(B)\le\widetilde\sigma_i(B)\le(1+\epsilon)\sigma_i(B)\),
using one row query to \(B\).  Translating again, each query costs
\(O(r)\) calls to \(\cO_A\), \(O(1)\)  call to \(\cO_w\), and
\(\widetilde O(r/\epsilon^{2})\) time.
\end{proof}

The following proposition describes the weight update step in the algorithm QMLSO (\cref{alg:qmlso}).

\begin{proposition}[Weight Computation for Overestimates]\label{prop:weight-computation}
 Let $s>0$. Assume query access $\cO_\cF$ to the function family $\cF=\sets{f_1,\ldots,f_m}$ and query access $\cO_w$ to a vector $w\in \RR^m_+$. Let $\cL$ denote an instance of $\mathsf{ModLevApprox}$. Then, there exists a quantum algorithm $\mathsf{WeightCompute}\parens{\cO_\cF,\cO_w,\cL,s}$ such that 
	\begin{equation*}
		\mathsf{WeightCompute}\parens{\cO_\cF,\cO_w,\cL,s}\ket{i}\ket{0}=\ket{i}\ket{w_i^\prime}
	\end{equation*}
	where
	\begin{equation*}
		w_i^\prime =\frac{1}{s}\cdot \frac{f_i\parens{\sqrt{\sigma_i/w_i}}}{\sigma_i/w_i},
	\end{equation*}
	and $\sigma_i$ is the leverage score returned by querying $\cL$ at index $i \in \sqb{m}$. The algorithm makes $O(1)$ queries to $\cL$, $\cO_w$, and $\cO_\cF$, and requires an additional $\widetilde{O}(1)$ time for computation.
\end{proposition}
\begin{proof}[Proof of \cref{prop:weight-computation}]
  Consider the following procedure:
  \[
  \begin{aligned}
    \ket{i}\ket{0}
    &\xmapsto{\cL.\mathsf{Query}}
      \ket{i}\ket{\sigma_i}\ket{0}\\
    &\xmapsto{\cO_w}
      \ket{i}\ket{\sigma_i}\ket{w_i}\ket{0}\\
    &\xmapsto{U_{\textsc{div}}}
      \ket{i}\ket{\sigma_i}\ket{w_i}\ket{\sigma_i/w_i}\ket{0}\ket{0}\\
    &\xmapsto{U_{\textsc{sqrt}}}
      \ket{i}\ket{\sigma_i}\ket{w_i}\ket{\sigma_i/w_i}\ket{\sqrt{\sigma_i/w_i}\ }\ket{0}\\
    &\xmapsto{\cO_\cF}
      \ket{i}\ket{\sigma_i}\ket{w_i}\ket{\sigma_i/w_i}\ket{\sqrt{\sigma_i/w_i}\ }\ket{f_i(\sqrt{\sigma_i/w_i}\ )}\\
    &\xmapsto{U_{\textsc{div}(sr)}}
      \ket{i}\ket{0}\ket{0}\ket{0}\ket{0}\ket{w_i'},
  \end{aligned}
  \]
  where \(U_{\textsc{div}}\) computes the division \(\sigma_i/w_i\),
  \(U_{\textsc{sqrt}}\) computes square root, and
  \(U_{\textsc{div}(sr)}\) carries out the arithmetic
  \(w_i':=\frac1s\cdot\frac{f_i(\sqrt{\sigma_i/w_i}\ )}{r}\) and then uncomputes the ancillary
  registers \(\sigma_i,w_i,\sigma_i/w_i,\sqrt{\sigma_i/w_i}\ \).
  
The computation uses $O(1)$ queries to $\cL$, $\cO_w$, and $\cO_\cF$, and  $\widetilde O(1)$ additional arithmetic operations, as claimed in the proposition.
\end{proof}

The following proposition specifies the output data structure produced by QMLSO (\cref{alg:qmlso}).
 
\begin{proposition}[Preparation for Multiscale Overestimates]\label{prop:preparation-multiscale-overestimates}
Let $0<\epsilon<1$, $\cJ = \sets{j_{\min}, \ldots, j_{\max}} \subseteq \ZZ$ be a contiguous interval. Suppose we are given a collection of instances $\sets{\cL^{(j)} : j \in \cJ}$, where each $\cL^{(j)}$ is an instance of $\mathsf{ModLevApprox}$. Then, there exists a quantum data structure $\mathsf{QOverestimate}$ that supports the following operations:
	\begin{itemize}
		\item Initialization: $\mathsf{QOverestimate}\parens{\sets{ \cL ^{(j)}: j\in \cJ},\epsilon}$, outputs an instance
		      $\mathcal{Z}$ in total time $ \widetilde O\parens{\sum_{j \in \cJ} \zeta _{j} }$,
		      where $\zeta _{j}$ denotes the needed time of initialization of $\cL^{(j)}$.
		\item Query: $\cZ.\mathsf{Query}$, outputs an unitary such that for every $i \in \sqb{m}$,
		      \begin{equation*}
			      \cZ.\mathsf{Query}\ket{i}\ket{0}=\ket{i}\ket{z _i},
		      \end{equation*}
		      where
		      \begin{equation*}
			      z _i = \frac{1}{1-\epsilon}  \sum_{ j \in \cJ}  \sigma ^{(j)}_i,
		      \end{equation*}
		      and $\sigma ^{(j)}_i$ is leverage score returned by querying $\cL^{(j)}$ at index $i\in \sqb{m}$. This query is executed  in $ \widetilde O\parens{\sum_{j\in \cJ}  \iota_j }$ time,
		      where $\iota_j$ represents the time required to querying $\cL^{(j)}$.
	\end{itemize}
\end{proposition}

\begin{proof}[Proof of \cref{prop:preparation-multiscale-overestimates}]
For every \(j \in \cJ\), suppose we hold an instance \(\cL^{(j)}\) of $\mathsf{ModLevApprox}$. Denote by \(\zeta_j\) its initialization cost and by \(\iota_j\) the cost of a single query. The initialization of $\mathsf{QOverestimate}$ consists of initializing all instances \(\cL^{(j)}\) for \(j \in \cJ\) and storing $\epsilon$ in QRAM in $\widetilde O \parens{1}$ time, allowing access through a unitary $U_\epsilon: \ket{i}\ket{0}\to \ket{i}\ket{1/1-\epsilon}$.

For the query operation, on input \(\ket{i}\ket{0}\), attach an accumulator register initialized to \(\ket{0}\), and iterate over the scales \(j \in \cJ\):
\[
\begin{aligned}
\ket{i}\ket{0} &
\xmapsto{\cL^{(j_{\min})}.\mathsf{Query}}     
  \ket{i}\ket{\sigma^{(j_{\min})}_i}\ket{0} 
\xmapsto{U_{\textsc{add}}}                           
  \ket{i}\ket{\sigma^{(j_{\min})}_i}\ket{\sigma^{(j_{\min})}_i} \ket{0}
  \\&\xmapsto{(\cL^{(j_{\min})}.\mathsf{Query})^{\dagger}}     
  \ket{i}\ket{0}\ket{\sigma^{(j_{\min})}_i}\ket{0} 
\\&\xmapsto{\cL^{(j_{\min}+1)}.\mathsf{Query}}      
  \ket{i}\ket{\sigma^{(j_{\min}+1)}_i}\ket{\sigma^{(j_{\min})}} \ket{0}
\xmapsto{U_{\textsc{add}}}                           
  \ket{i}\ket{\sigma^{(j_{\min}+1)}_i}\ket{\sigma^{(j_{\min})}+\sigma^{(j_{\min}+1)}_i} \ket{0}
  \\&\xmapsto{(\cL^{(j_{\min}+1)}.\mathsf{Query})^{\dagger}}  
  \ket{i}\ket{0}\ket{\sigma _i^{(j_{\min})}+\sigma _i^{(j_{\min}+1)}}\ket{0}
\\&\xmapsto{\quad\cdots\quad}  
  \ket{i}\ket{0}\ket{\sum_{j\in \cJ}\sigma _i^{(j)}}\ket{0}
  \\&
  \xmapsto{ U _{\epsilon}}  
  \ket{i}\ket{1/\parens{1-\epsilon}}\ket{\sum_{j\in \cJ}\sigma _i^{(j)}}\ket{0}
    \xmapsto{ U _{\textsc{mult}}}  
  \ket{i}\ket{1/\parens{1-\epsilon}}\ket{\sum_{j\in \cJ}\sigma _i^{(j)}}\ket{z_i},
\end{aligned}
\]
The adder \(U_{\textsc{add}}\) and multiplier $U _{\textsc{mult}}$  are implemented using fixed-point arithmetic and costs \(\widetilde{O}(1)\) gates.
The query procedure calls each \(\cL^{(j)}\) once and invokes its inverse once; together this costs \(\Theta(\iota_j)\). Summing over all \(j \in \cJ\), the total runtime is \(O\bigl(\sum_{j \in \cJ} \iota_j\bigr)\). The arithmetic overhead from the adders contributes \(\widetilde{O}(|\cJ|)\) additional gates, which is absorbed in the stated complexity.
\end{proof}

\subsection{Proof of \cref{thm:qmlso}}

We first analyze the time complexity. 
Since $\theta\in (0,4)$,  we have $\delta= \max\set{\frac{1}{2}, \abss{\frac{\theta-2}{2}}}<1$, and thus $T=\Theta_{L,\theta,c}\parens{\log\log \beta}$. In each iteration $i\leq T$, the subroutine $\mathsf{ModLevApprox}$ runs in $\widetilde O \parens{r \sqrt{mn}+n ^\omega+n r ^2}$ time (absorbing the cost of queries), and provides query access such that  each query requires $O (r)$ calls to $\cO _A$, $O (1)$ calls to $\cO_{w^{(i)}_{\text{itr}}}$, and $\widetilde O (r)$ time (\cref{prop:modlevapprox}). The subroutine  $\mathsf{WeightCompute}$ then constructs $U_{w^{(i+1)}_{\text{itr}}}$, where each query  requires $\widetilde O \parens{r}$ time.  This is because it makes $O(1)$ calls to $\cL_{\text{itr}}^{(i)}$ and $\cO_{w^{(i)}_{\text{itr}}}$, each of which can be executed in $\widetilde{O}(r)$ time (\cref{prop:weight-computation}).
For the downward recursion across $\abs{\cJ}$ scales (lines 9-12 of \cref{alg:qmlso}), the complexity is the same per step. Therefore, the total cost of iterations and recursion is $\widetilde O\paren{ \parens{r \sqrt{mn}+n^\omega+n r ^2}\Delta}$, makeing $O\parens{\sqrt{mn}}$ queries to $\cO_{w^\circ}$, $O \parens{r\sqrt{mn}\Delta}$ queries to $\cO_{A}$, $O \parens{\sqrt{mn}\Delta}$ queries to $\cO_{\cF}$, where $\Delta=\abss{j_{\max}-j_{\min}}+\max\sets{\log \log\beta,0}$.
In the final step of the algorithm, the data structure $\cZ$ can be initialized in
  $\widetilde O \paren{\parens {r \sqrt{mn}+n^\omega+nr^2}\abss{\cJ}}$ time, and each call to  $\cZ.\mathsf{Query}$ can be executed
  in $\widetilde O\paren{r\abss{\cJ}} $ time. This follows directly from 
  \cref{prop:preparation-multiscale-overestimates}, with
  parameters $\zeta_t =\widetilde O\parens {r \sqrt{mn}+n^\omega+nr^2}$   and $\iota_t = \widetilde O\parens {r}$.

In the remainder of this section, we will show that the vector $z\in \RR^m$ corresponding to the output of QMLSO (\cref{alg:qmlso}) is indeed an 
$O \parens{n\cdot\abss{j_{\max}-j_{\min}}}$-overestimate with respect to  $A$ and an $\alpha$-approximate scheme $\cW$.
To analyze the correctness of the algorithm, we first introduce the following metric.  
\begin{definition}
	Define a metric $d$ on $\RR^m_+$ by
\begin{equation*}
	d \parens{u,w}:= \max\set{\abss{\log \frac{u_i}{w_i}}: i \in \sqb{m}}
\end{equation*}
\end{definition}
It is straightforward to verify that $d$ is indeed a metric, as it satisfies symmetry, positivity, and the triangle inequality.

A useful characterization is the following, which is immediate from the definition: it shows that being an
$\alpha$-approximate weight is equivalent to requiring that $w$ is
almost a fixed point of the update map $\phi_s$, up to distance
$\log \alpha$ in the metric $d$.
\begin{fact}\label{fa:approximate-weight}
Fix a matrix $A\in \RR^{m\times n}$ (where the $i$-th row is $a_i$ for each $i\in \sqb{m}$) and a family of loss functions $\cF=\sets{f_1,\ldots,f_m}$. A vector $w \in \RR^m _+$ is an $\alpha$-approximate weight  (with respect to $A$ and $\cF$)
 at scale $s$ if and only if 
	\begin{equation*}
		d \parens{w, \phi_s \parens{w}} \leq \log \alpha
	\end{equation*}
	where $\phi_s:\RR^m_+\to \RR^m_+$ is the update function defined by
\begin{equation*}
		\parens{\phi_s\parens{ w}}_i :=\frac{1}{s}\cdot \frac{f_i\parens{\sqrt{\sigma_i\parens{W^{1/2}A}/w_i}}}{\sigma_i\parens{W^{1/2}A}/w_i},\quad \text{and}\quad W:=\diag\paren{w}.
\end{equation*}
\end{fact}

Furthermore,  $\cW=\sets{w ^{(j)}\in \RR^m_+: j \in \cJ}$ is an $\alpha$-approximate weight scheme if and only if
\begin{align}\label{eq:weight-scheme-equality-1}
	\forall j\in\cJ:\ 
d\bigl(w^{(j)},\phi_{2^{j}}\bigl(w^{(j)}\bigr)\bigr)\leq \log\alpha;
\\\label{eq:weight-scheme-equality-2}
\forall j,j+1\in\cJ:\ d \bigl(w^{(j+1)},w^{(j)}\bigr)\leq \log \alpha.
\end{align}

In the noisy setting (where \cref{alg:qmlso} actually operates), we define the update function $\widetilde{\phi}_{s,\epsilon} : \mathbb{R}_{+}^m \to \mathbb{R}_{+}^m$ by
\begin{equation}\label{eq:noisy-update-function}
\paren{\widetilde \phi_{s, \epsilon} \parens{w}}_i:=\frac{1}{s} \frac{f_i\left(\sqrt{\tilde{\sigma}_i / w_i}\right)}{\tilde{\sigma}_i / w_i},\quad \forall i \in\sqb{m}, 
\end{equation}
where $\widetilde{\sigma}_i $ is an $\epsilon$-approximation of $\sigma_i\parens{W^{1/2}A}$, i.e., $(1-\epsilon)\sigma_i \parens{W^{1/2}A}\leq \widetilde{\sigma}_i \leq (1+\epsilon)\sigma_i\parens{W^{1/2}A}$. 

Our analysis builds on the ideas of \citet{jambulapati2024sparsifying}. We use the following result, which shows that the weight update procedure remains ``contractive'' even when noise is present.

\begin{lemma}[Contraction of a noisy update,{\cite[Lemma 2.9]{jambulapati2024sparsifying}}]\label{le:contraction-update}
Let $\delta=\max\set{\frac{1}{2}, \abss{\frac{\theta-2}{2}}}$, $C=\max\set{\frac{2L}{c} ,\frac{1}{c}}$.
For any $w,w ^\prime\in \RR^m _+$ and $\epsilon,\epsilon^\prime \in [0,1)$, it holds that 
$	d \parens{\widetilde \phi_{s,\epsilon} \parens{w},\widetilde \phi_{s,\epsilon^\prime}\parens{w^\prime}}\leq  \delta\cdot  d \parens{w ,w ^\prime}+\log C+ \log\frac{1+\epsilon}{1-\epsilon^\prime}$.
\end{lemma}
At the beginning, assuming the initial weight is $\beta$-approximate, we apply the triangle inequality of $d$ together with \cref{le:contraction-update} to get
\begin{align*}
	d \paren{w^{\parens{0}}_\text{itr} , \widetilde \phi _{s_{\max,\epsilon}}\paren{w^{\parens{0}}_\text{itr}}}&\leq 	d \paren{w^{\parens{0}}_\text{itr} ,  \phi _{s_{\max}}\paren{w^{\parens{0}}_\text{itr}}}+	d \paren{\phi _{s_{\max}}\paren{w^{\parens{0}}_\text{itr}}, \widetilde \phi _{s_{\max,\epsilon}}\paren{w^{\parens{0}}_\text{itr}}}\\
	&\leq\log \beta+\log C +\log \parens{1+\epsilon}.
\end{align*}

For the iterative procedure (lines 4-7 of \cref{alg:qmlso}), by applying \cref{le:contraction-update} inductively, we obtain
\begin{align*}
	d \paren{w^{\parens{T}}_\text{itr} , \widetilde \phi _{s_{\max,\epsilon}}\paren{w^{\parens{T}}_\text{itr}}}&\leq \delta^T	\cdot d \paren{w^{\parens{0}}_\text{itr} ,  \widetilde\phi _{s_{\max},\epsilon}\paren{w^{\parens{0}}_\text{itr}}}+	\sum_{i=0}^{T-1}\delta ^i \paren{\log C +\log \frac{1+\epsilon}{1-\epsilon}}
	\\&\leq  \delta^T	\cdot \log \beta+\sum_{i=0}^{T}\delta ^i \paren{\log C +\log \frac{1+\epsilon}{1-\epsilon}}
	\\&\leq  \frac{\delta ^T \log \beta +\log C +\log \frac{1+\epsilon}{1-\epsilon}}{1-\delta }.
\end{align*}
Therefore, after $T$ iterations, we have
\begin{align*}
		d \paren{w^{\parens{T}}_\text{itr} , \phi _{s_{\max}}\paren{w^{\parens{T}}_\text{itr}}}
		&\leq 	d \paren{w^{\parens{T}}_\text{itr} , \widetilde \phi _{s_{\max},\epsilon}\paren{w^{\parens{T}}_\text{itr}}}+	d \paren{\widetilde \phi _{s_{\max},\epsilon}\paren{w^{\parens{T}}_\text{itr}} , \widetilde\phi _{s_{\max},0}\paren{w^{\parens{T}}_\text{itr}}}\\
		&\leq d \paren{w^{\parens{T}}_\text{itr} , \widetilde \phi _{s_{\max},\epsilon}\paren{w^{\parens{T}}_\text{itr}}}+\log C+\log \parens{1+\epsilon}
		\\&\leq \frac{\delta ^T \log \beta +\log C +\log \frac{1+\epsilon}{1-\epsilon}}{1-\delta }+\log C+\log \parens{1+\epsilon}
		\\&\leq \frac{\delta ^T \log \beta +2\log C +2\log \parens{1+\epsilon}-\log\parens{1-\epsilon}}{1-\delta }
		\\&\leq \frac{3 \log  C }{1-\delta },
\end{align*}
where 
the first inequality follows from the triangle inequality of $d$, 
the second inequality applies \cref{le:contraction-update}, 
the third inequality substitutes the previous bound, 
and the last inequality follows from the choice of $T$ (line~2 of \cref{alg:qmlso}).

Set
$\alpha:= C^{3/\parens{1-\delta}} 
  = \Theta_{L,\theta,c}(1)$.
By \cref{fa:approximate-weight}, we have that $w^{(j_{\max})} = w^{(T)}_{\text{itr}}$ is $\alpha$-approximate at scale $s_{\max} = 2^{j_{\max}}$.

We now proceed with a downward recursion across scales (lines 9-12 of \cref{alg:qmlso}).
For each $j=j_{\max} -1,\dots,j_{\min}$, define
\(
w^{(j)}:=\widetilde\phi_{2^{j},\epsilon}\bigl(w^{(j+1)}\bigr).
\) By applying \cref{le:contraction-update} inductively, we obtain that
\begin{align*}
& \ d \paren{w^{(j)},  \widetilde\phi_{2^{j},\epsilon}\paren{w^{(j)}}}\\&= 
d \paren{\widetilde\phi_{2^{j},\epsilon}\paren{w^{(j+1)}}, \widetilde\phi_{2^{j},\epsilon}^2\paren{w^{(j+1)}}}
\\&
\leq \delta	\cdot d \paren{w^{\parens{j+1}} ,  \widetilde\phi _{2^{j},\epsilon}\paren{w^{\parens{j+1}}}}+	\log C +\log \frac{1+\epsilon}{1-\epsilon}
\\&=
\delta	\cdot \paren{d \paren{w^{\parens{j+1}} ,  \widetilde\phi _{2^{j+1},\epsilon}\paren{w^{\parens{j+1}}}}+ \log 2} +	\log C +\log \frac{1+\epsilon}{1-\epsilon} 
\\&
\leq \delta^{j_{\max}-j}	\cdot \paren{d \paren{w^{\parens{j_{\max}}},  \widetilde\phi _{s_{\max},\epsilon}\paren{w^{\parens{j_{\max}}}}}+\log 2}+	\sum_{i=0}^{j_{\max}-1-j}\delta ^i \paren{\log C +\log \frac{1+\epsilon}{1-\epsilon}}
	\\&
	\leq \log \alpha+\log 2+\frac{ \log C +\log \frac{1+\epsilon}{1-\epsilon}}{1-\delta }.
\end{align*}
Using the above result, we can now verify that \cref{eq:weight-scheme-equality-1} holds.
\begin{align*}
d \paren{w^{(j)},  \phi_{2^{j}}\paren{w^{(j)}}}&=d \paren{w^{(j)},  \widetilde\phi_{2^{j},\epsilon}\paren{w^{(j)}}}+
d \paren{\widetilde\phi_{2^{j},\epsilon}\paren{w^{(j)}}, \widetilde\phi_{2^{j},0}\paren{w^{(j)}}}
\\&\leq 
d \paren{w^{(j)},  \widetilde\phi_{2^{j},\epsilon}\paren{w^{(j)}}}+\log C+\log \parens{1+\epsilon}
\\&
\leq \log \alpha+\log 2+\frac{ \log C +\log \frac{1+\epsilon}{1-\epsilon}}{1-\delta }+\log C+\log \parens{1+\epsilon}
\\&\leq \log  \alpha +\log 2+\frac{ 2\log C +2\log \parens{1+\epsilon}-\log\parens{1-\epsilon}}{1-\delta }
\\&\leq2\log \alpha+\log 2.
\end{align*}
And confirm that \cref{eq:weight-scheme-equality-2} is satisfied.
\begin{equation*}
	d \paren{w^{(j+1)}, w ^{(j)}}=
d \paren{w^{(j+1)}, \widetilde \phi_{2^{j},\epsilon}\paren{w^{(j+1)}}}
=d \paren{w^{(j+1)}, \widetilde \phi_{2^{j+1},\epsilon}\paren{w^{(j+1)}}}+\log 2
\leq 2 \log 2\alpha. 
\end{equation*}
Therefore, $\cW := \set{w^{(j)} : j  \in \cJ:= \ZZ\cap [j_{\min},j_{\max}]}$ forms a $4\alpha^2$-approximate weight scheme. It remains to verify that the vector $z$ in line 12 of \cref{alg:qmlso} (see \cref{prop:preparation-multiscale-overestimates} for details) is an $ O \parens{n \cdot\abss{j_{\max}-j_{\min}}}$-overestimate with respect to $\cW$.

For every $i\in \sqb{m}$, note that
\begin{equation*}
	z_{i}=\frac{1}{1-\epsilon}\sum_{j\in \cJ}\widetilde \sigma_i \paren{W^{1/2}_j A}\geq \sum_{j\in \cJ} \sigma_i \paren{W^{1/2}_j A}
	\geq \max_{j\in \cJ}\sigma_i \parens{W^{1/2}_j A},
\end{equation*}
and
\begin{align*}
\norms{z}_1&=\sum_{i\in \sqb{m}} z_i=\frac{1}{1-\epsilon}\sum_{i\in \sqb{m}}\sum_{j \in \cJ}\widetilde \sigma_i \paren{W^{1/2}_j A}\leq \frac{1+\epsilon}{1-\epsilon}\sum_{j\in \cJ}\sum_{i \in \sqb{m}}\sigma_i \paren{W^{1/2}_j A}\\&\leq \frac{1+\epsilon}{1-\epsilon}\sum_{j\in \cJ}\trace\parens{W_j^{1/2}A\paren{A^\top W_jA}^+ (W^{1/2}_jA)^\top  }\leq\frac{1+\epsilon}{1-\epsilon}\sum_{j\in \cJ}\rank \parens{A}
\\&
\leq \frac{1+\epsilon}{1-\epsilon}\parens{j_{\max}-j_{\min}}n.
\end{align*}

\section{Proof of~\cref{thm:weight-initialize}}\label{apdx:weight-initialize}

\begin{proof}[Proof of~\cref{thm:weight-initialize}] 
The construction follows the same three–step recipe as in the classical
argument of \citet{jambulapati2024sparsifying}:  
(1) pick anchor points $\hat t_i$ on each loss function;  
(2) use $\mathsf{ModLevApprox}$ to estimate leverage scores at those anchors and
convert them into weights;  
(3) add a quadratic bump to every $f_i$ to obtain a better-behaved family
$\cF^\circ$ together with an approximate weight vector~$w^\circ$.

\textbf{Step 1: Anchor Selection.}  
Since each $f_i$ is $(L,\theta,c)$-proper, we can efficiently search over $t>0$ to locate anchor points $\widehat t_i$ satisfying
\begin{equation*}
\frac{1}{2}s_{\max}\leq f _i \parens{\widehat t_i}\leq s_{\max},\quad i \in \sqb{m},
\end{equation*}
as detailed in \cref{prop:weight-init-search}. Thus, a quantum oracle $\cO_t$ to these anchors can be constructed using $ O(1)$ queries to $\cO_{\cF}$.

\textbf{Step 2: Quantum Leverage Score Approximation.} 
Invoke $\mathsf{ModLevApprox}(\cO_{A}, \cO_{h}, \frac{1}{2})$, where $\cO_{h}$ queries the vector $h=(\widehat t_1^{-2},\dots,\widehat t_m^{-2})$, constructed in $\widetilde O(1)$ time via $O(1)$ queries to $\cO_t$. The quantum algorithm uses $\widetilde O(r\sqrt{mn})$ queries to $\cO_A$, $\widetilde O(\sqrt{mn})$ queries to $\cO_\cF$, and runs in $\widetilde O(r\sqrt{mn}+n^\omega+nr^2)$ time. It outputs an approximation $\widetilde\sigma$ satisfying $\frac{1}{2}\sigma(H^{1/2}A)\leq\widetilde\sigma_i\leq\frac{3}{2}\sigma(H^{1/2}A)$ for $H=\text{diag}(h)$. Set
\begin{equation*}
	w_i^\circ:=\delta/ \widetilde \sigma_i, \quad i \in \sqb{m},
\end{equation*}
and define query access $\cO_{w^{\circ}}$, where each query requires $O(r)$ queries to $\cO_A$, $O(1)$ queries to $\cO_{\cF}$, and takes $\widetilde O(r)$ time.

\textbf{Step 3: Modified Loss Definition.}
Define modified losses as:
\begin{equation*}
f_i^\circ\parens{t}= f_i \parens{t}+s_{\max} w_i ^\circ\cdot t ^2, \quad\forall  i \in \sqb{m}.
\end{equation*}

Following~\citet[Fact 2.10]{jambulapati2024sparsifying}, each $(f_i^\circ)^{1/2}$ is $\max\sets{1,L}$-auto-Lipschitz, $\theta$-lower-homogeneous with constant $c$, making $\cF^\circ$ a $(\max\sets{1,L},\theta,c)$-proper family. The correctness of \cref{eq:weight-initialize} and that $w^\circ$ is $\beta$-approximate at scale $s_{\max}$ directly follows from~\citet[Theorem 2.11]{jambulapati2024sparsifying}, yielding $C_{\textsc{init}}=2(2L/c)^{2/\theta}$.
\end{proof}

\begin{proposition}\label{prop:weight-init-search}
Let $\cO_f$ be an oracle for a $\parens{L,\theta,c}$-proper function $f:\RR \to \RR_+$. For any $s_{\max} > s_{\min} > 0$, there exists an algorithm 
$\mathsf{FindSol}(\cO_f, s_{\min}, s_{\max})$
that outputs $x>0$ such that $s_{\min} \leq f(x) \leq s_{\max}$, making at most
$O_{L,\theta,c}\left(\log\frac{s_{\max}}{s_{\min}}\right)$
queries to $\cO_f$, and running in $
O_{L,\theta,c}\left(\log\frac{s_{\max}}{s_{\min}}\right)$ time.
\end{proposition}

\begin{proof}[Proof of \cref{prop:weight-init-search}]
Let $h=\sqrt{f}$. Recall \cref{def:proper-loss}, $h$ is $L$-auto-Lipschitz and lower $\theta$-homogeneous with constant $c>0$.
We implement oracle access to $h$ via a single call to $\cO_f$ and a  square root, so the query complexity to $\cO_h$ equals that to $\cO_f$. Let $h_{\min}:=\sqrt{s_{\min}}, h_{\max}:=\sqrt{s_{\max}}$. Our goal is to find $x > 0$ such that $h_{\min} \leq h(x) \leq h_{\max}$.

We first query $h(1)$. If $h_{\min} \leq h(1) \leq h_{\max}$, we are done.  
Otherwise, we use exponential search to find an interval $[a,b]$ (with $b=2a$ or $a=2b$) such that $h$ crosses the range $[h_{\min},h_{\max}]$ within this bracket.

\paragraph{Case 1: $h(1)<h_{\min}$.}  
By lower $\theta$-homogeneity,
\[
h(2^j) \geq c \cdot 2^{j\theta} h(1), \qquad j \geq 1.
\]
Thus, for 
$j^\star = \Big\lceil \tfrac{1}{\theta}\log_2 \tfrac{h_{\max}}{c\cdot  h(1) } \Big\rceil$,
we have $h(2^{j^\star}) \geq h_{\max}$. Therefore, within $O(\tfrac{1}{\theta}\log\frac{h_{\max}}{c\cdot h(1)})$ queries we either find $x$ with $h(x)\in[h_{\min},h_{\max}]$ or obtain a bracket $[2^{j^\star-1}, 2^{j^\star}]$.

\paragraph{Case 2: $h(1)>h_{\max}$.}  
For $j\geq 1$, taking $\lambda=2^j$ and $x=2^{-j}$ in the homogeneity property gives
\[
h(1) = h(2^j \cdot 2^{-j}) \geq c\cdot 2^{j\theta} h(2^{-j})
\quad\Rightarrow\quad
h(2^{-j}) \leq \tfrac{h(1)}{c \cdot 2^{j\theta}}.
\]
Hence for $j^\star = \Big\lceil \tfrac{1}{\theta}\log_2 \tfrac{h(1)}{c\cdot h_{\min}} \Big\rceil$,
we have $h(2^{-j^\star}) \leq h_{\min}$, and again we either hit the target range or obtain a bracket $[2^{-j^\star}, 2^{-(j^\star-1)}]$.

\paragraph{Bisection step.}	
Assume w.l.o.g.\ we have a bracket $[a,b]$ with $b=2a$ and
$h(a)\leq h_{\min}<h_{\max}\leq h(b)$ (the mirrored case is analogous).
Let $\Delta:=h_{\max}-h_{\min}>0$ and let the half-length be $\delta_0:=\tfrac{b-a}{2}=\tfrac{a}{2}$.
For $t\geq 0$, let the half-length after $t$ bisections be $\delta_t := \delta_0/{2^{t}}$.
By lower $\theta$-homogeneity (applied with $x=\delta_{t+1}$ and $\lambda=2$),
\[
h(\delta_{t+1}) \leq \frac{1}{c2^\theta}h(\delta_t)
\quad\Rightarrow\quad
h(\delta_t) \leq \Bigl(\frac{1}{c2^\theta}\Bigr)^{t} h(\delta_0).
\]
Moreover, since $b-a=a$ and $h(a)\leq h_{\min}$, we have
\[
h(\delta_0) = h\Bigl(\frac{a}{2}\Bigr) \leq \frac{1}{c2^\theta} h(a)
\leq \frac{h_{\min}}{c2^\theta}.
\]
Combining the two displays yields
\[
Lh(\delta_t)
\leq
L\Bigl(\frac{1}{c2^\theta}\Bigr)^{t+1} h_{\min}.
\]
Choose
\[
t \geq
\biggl\lceil
\frac{1}{\theta+\log_2 c}
\log_2\Bigl(\frac{Lh_{\min}}{\Delta}\Bigr)
\biggr\rceil-1.
\]
Then $L h(\delta_t)\leq \Delta$. Let $m$ be the midpoint of the current bracket $[a^\prime, b^\prime ]$ (of length $\delta_{t-1}$).
By the $L$-auto-Lipschitz property,
\begin{align*}
h(m) &\leq h(a^\prime)+L h(\delta_t)\leq h_{\min}+\Delta = h_{\max},
\\
h(m) &\geq h(b^\prime)-L h(\delta_t) \geq h_{\max}-\Delta = h_{\min},
\end{align*}
so $h(m)\in [h_{\min},h_{\max}]$ and the midpoint succeeds.

Overall, the total query complexity of $\mathsf{FindSol}$ is 
$O_{L,\theta,c}\left(\log\frac{s_{\max}}{s_{\min}}\right)$,  under the assumption that $f(1)$ is bounded by $\poly(s_{\min},s_{\max})$.
The time is linear in the number of queries.

\end{proof}

\section{Proof of \cref{thm:quantum-glm-sparsification}}\label{apdx:glm-sparsification}

Before proving \cref{thm:quantum-glm-sparsification}, we restate two quantum algorithms that are invoked as subroutines in our algorithm.

\begin{theorem}[Multiple Quantum State Preparation, Theorem 1 in~\cite{hamoudi2022preparing}]\label{thm:quantum-prob-sample}
Let $1\leq k\leq n$.
	There exists a quantum algorithm $\mathsf{MultiSample}(\cO _w, k)$ that, given query access $\cO_w$ to a nonnegative vector $w \in \RR _{+}^n$, outputs with high probability a sample sequence $\vartheta\in \sqb{n}^k$ in which each element $i$ is drawn with probability proportional to $w _i$. The algorithm makes $\widetilde O\parens {\sqrt{kn}}$ queries to $\cO_w$ and runs in  $\widetilde O\parens {\sqrt{kn}}$ time.
\end{theorem}

\begin{theorem}[Quantum Sum Estimation, {\cite[Lemma~3.1]{li2019sublinear}}]\label{thm:quantum-sum-estimate}
There exists a quantum algorithm $\mathsf{SumEstimate}(\cO_w,\epsilon)$ that, given query access $\cO_w$ to a nonnegative vector $w \in \RR_{+}^n$, outputs with high probability an $\epsilon$-approximation $\widetilde{s}$ to the sum $s=\sum_{i\in \sqb{n}} w_i$. The algorithm uses $\widetilde{O}(\sqrt{n}/\epsilon)$ queries to $\cO_w$ and runs in $\widetilde{O}(\sqrt{n}/\epsilon)$ time.
\end{theorem}

We will use the following result from~\cite{jambulapati2024sparsifying} in a black-box manner.
\begin{theorem}\label{thm:glm-main}
Let $F(x) := \sum_{i=1}^m f_i\left(\langle a_i, x \rangle \right)$ be a total loss function defined by vectors $a_1, \ldots, a_m \in \mathbb{R}^n$ and a $\parens{L,\theta,c}$-proper loss functions family $\cF =\sets{f_1, \ldots, f_m }$.
Let $0\leq s_{\min}<s_{\max}$. 
Given any $O_{L,\theta,c}(1)$-approximate weight scheme $\cW$ with respect to matrix $A$ (with rows $a_i$), loss family $\cF$, and interval $\cJ := [\log_2 s_{\min}-4\log m,\log_2 s_{\max}] \bigcap \mathbb{Z}$, 
with high probability,
importance sampling each function $f_i \in \cF$ with probability
\begin{equation*}
\quad p_i \geq \max_{j \in \cJ}\sigma_i(W_j^{1/2}A)
\end{equation*}
yields an $s$-sparse $\epsilon$-approximate sparsifier of $F$ over the range $[s_{\min}, s_{\max}]$, where
\begin{equation*}
s=O \parens{\norm{p}_1 \log^3 m/\epsilon^2 }.
\end{equation*}
\end{theorem}

\begin{proof}[Proof of \cref{thm:quantum-glm-sparsification}]

According to \cref{thm:weight-initialize}, we obtain query access to $\cF^\circ$ and ${w^\circ}$ in $\widetilde O\parens{r \sqrt{mn}+n ^\omega+n r^2}$ time, and each query costs $\widetilde O\parens{r}$ time. These satisfy the property that $w^\circ$ is an $O \parens{m^4s_{\max} /\parens{\epsilon s_{\min}}}$-approximate weight at scale $s_{\max}$ with respect to $\cF ^\circ$, and moreover $\cF^\circ$ satisfies the condition
\begin{equation}\label{eq:quantum-glm-proof-1}
		\sum_{i=1}^m f_i \paren{\ips{a_i}{x}}\leq s_{\max}\quad\Rightarrow \quad 0\leq \sum_{i=1}^m f_i^\circ \paren{\ips{a_i}{x}}-f_i \paren{\ips{a_i}{x}}\leq\frac{ C_{\textsc{init}} \epsilon s_{\min} }{ m}.
\end{equation}

Note that $\abss{\cJ}=j_{\max}-j_{\min}=O\parens{\log \parens{m s_{\max}/s_{\min}}}=\widetilde O\parens{\log \parens{s_{\max}/s_{\min}}}$.
By applying  \cref{thm:qmlso}, we obtain query access to a multiscale $O \parens{n\log \parens{s_{\max}/s_{\min}}}$-overestimates $z$ with respect to $\cF^\circ$. This construction requires $\widetilde O \paren{\parens{r\sqrt{mn}+n ^\omega+n r^2}\log \parens{s_{\max}/s_{\min}}}$ time, and each query can be executed in $\widetilde O \parens{r \log \parens{s_{\max}/s_{\min}}}$ time.
Combining this with the multiple quantum state preparation (\cref{thm:quantum-prob-sample}), we can generate a sequence $\vartheta \in  \sqb{m}^M$ in
  $\widetilde O\parens{\sqrt{Mm} \cdot \log \parens{s_{\max}/s_{\min}}} =\widetilde O\parens{r \sqrt{mn}  \log \parens{s_{\max}/s_{\min}}/\epsilon}$ time.
The quantum sum estimation procedure (\cref{thm:quantum-sum-estimate}) then runs in 
  $\widetilde O\parens{r \sqrt{m} \log \parens{s_{\max}/s_{\min}}}$ time, producing $\widetilde \nu$ as a
  $0.1$-approximation of $\nu=\norms{z}_1$. Therefore, the overall time complexity is $\widetilde O \paren{\parens{n ^\omega+nr^2 +r\sqrt{mn}/\epsilon }\cdot {\log \parens{s_{\max}/s_{\min}}}}$. 

It remains to show that
\begin{equation*}
\widetilde F(x):=\sum_{i=1}^m \frac{\#\sets{k\in \sqb{M}: \vartheta_k=i}}{1.1M z_i/\widetilde \nu} f_i\parens{\ips{a_i}{x}}	
\end{equation*}
 is an $\widetilde O \parens{ \log\parens{s_{\max}/s_{\min}}\cdot n/\epsilon^2}$-sparse $\epsilon$-approximate sparsifier of $F$ over the range $\sqbs{s_{\min},s_{\max}}$.

In each sampling round (out of $M$ total rounds), the probability of selecting each  $f_i\in \cF$  can be expressed (after reweighting) as
\begin{equation*}
	1.1 z_i\cdot \frac{\nu}{\widetilde \nu}\geq \sigma^{\multiscale}_i \parens{A, \cW^\circ}\geq \max_{j\in \cJ}\sigma \parens{{W^\circ}^{1/2}_j A},
\end{equation*}
where $W^\circ$ is an $O\parens{1}$-approximate weight scheme with respect to $\cF^\circ$. 
The inequalities follow from the bound $\widetilde \nu \leq 1.1 \nu$ and from \cref{def:mlso}.

Applying \cref{thm:glm-main}, we obtain that $\widetilde F(x)$ is an $O\parens{\norms{z}_1\log ^3m/\epsilon}=\widetilde O \parens{n\log\parens{s_{\max}/s_{\min}} /\epsilon^2}$-sparse $\epsilon$-approximate sparsifier of $F^\circ$ over the range $[s_{\min},s_{\max}]$. 
Finally, by \cref{eq:quantum-glm-proof-1}, for all $x$ such that $ s_{\min}\leq  F \parens{x}\leq s_{\max}$, we have
\begin{align*}
	\abss{\widetilde F\parens{x}- F\parens{x}}&\leq 	\abss{ \widetilde F\parens{x}- F^\circ\parens{x}}+	\abss{ F^\circ\parens{x}- F\parens{x}} \\
	&\overset{\cref{eq:quantum-glm-proof-1}}{\leq} \epsilon  F^\circ \parens{x} + \frac{C_{\textsc{init}}\epsilon s_{\min}}{m}\\
	&\overset{\cref{eq:quantum-glm-proof-1}}{\leq}  \epsilon\paren{ F \parens{x} +\frac{C_{\textsc{init}}\epsilon s_{\min }}{m} }+\frac{C_{\textsc{init}}\epsilon s_{\min} }{m}\\
		&\leq  \epsilon F \parens{x} +\frac{C_{\textsc{init}}\parens{\epsilon^2+\epsilon }}{m} F \parens{x} \\
	&\leq 2 \epsilon F\parens{x}.
\end{align*}
By setting $\epsilon \leftarrow \tfrac{\epsilon}{2}$, we obtain we  desired.
\end{proof}

\section{Applications}\label{apdx:applications}

As a special case of $\ell_p$ regression with $p=2$, the multiscale leverage score overestimates reduce to the  leverage score overestimates. In this case, we can directly apply quantum leverage score approximation (\cref{thm:quantum-leverage-score-aapproximation}) with constant accuracy, without invoking QMLSO (\cref{alg:qmlso}). Combining the classical linear regression algorithm~\citep[Theorem~43]{clarkson2017low} with our quantum sparsification framework yields the following corollary.

\begin{corollary}[Quantum Linear Regression]\label{cor:linear-regression}
There exists a quantum algorithm that, given query access to a matrix $A\in \RR^{m\times n}$ (with row sparsity $r \leq n$) and vector $b\in \RR^m$, and $\epsilon> 0$, outputs with high probability an $x\in \RR^n$ such that $\norms{Ax-b}_2 \leq \parens{1+\epsilon} \min _{x\in\RR^n }\norms{Ax-b}_2$, in $\widetilde O \parens{r\sqrt{mn}/\epsilon+n^3}$ time. 
\end{corollary}
Note that the sparsifier size is $m^\prime=\widetilde O(n/\epsilon^2)$, which is smaller than the original  size $m$,  implying $\epsilon = \Omega(\sqrt{n/m})$. Consequently, the classical runtime term $\widetilde O(mr)$ reduces to $\widetilde O(m^\prime r)=\widetilde O(nr/\epsilon^2)=\widetilde O(r\sqrt{mn}/\epsilon)$.

As discussed in \cref{subsec:application}, multiple regression can be reformulated as linear regression through vectorization. Accordingly, we first apply quantum leverage score approximation (\cref{thm:quantum-leverage-score-aapproximation}) to the matrix $A \in \mathbb{R}^{m\times n}$ and perform importance sampling to obtain a reduced matrix $A'$, of size $\widetilde O(n/\epsilon^2) \times n$. Finally, we apply the classical multiple regression algorithm (e.g.,~\citep[Theorem~38]{clarkson2017low}) to $A'$ to compute the solution.
\begin{corollary}[Quantum Multiple Regression]\label{cor:multiple-regression}
There exists a quantum algorithm that, given query access to a matrix $A\in \RR^{m\times n}$ (with row sparsity $r \leq n$) and matrix $B\in \RR^{m\times N}$, and $\epsilon> 0$, outputs with high probability an $X\in \RR^{n\times N}$ such that $\norms{AX-B}_F \leq \parens{1+\epsilon} \min _{x\in\RR^n }\norms{AX-B}_F$, in $\widetilde O \paren{r\sqrt{mn}/\epsilon+n^2\parens{Nn+n/\epsilon+N/\epsilon}}$ time. 
\end{corollary}
Ridge regression can be reduced to linear regression by augmenting with $n$ additional data points (see \cref{subsec:application}).
\begin{corollary}[Quantum Ridge Regression]\label{cor:ridge-regression}
Let $\lambda>0$ be a regularization parameter. There exists a quantum algorithm that, given query access to a matrix $A\in \RR^{m\times n}$ (with row sparsity $r \leq n$) and vector $b\in \RR^m$, and $\epsilon> 0$, outputs with high probability an $x\in \RR^n$ such that $\norms{Ax-b}_2^2+\lambda \norms{x}_2^2 \leq \parens{1+\epsilon} \min _{x\in\RR^n }\paren{\norms{Ax-b}_2^2+\lambda\norms{x}_2^2}$, in $\widetilde O \parens{r\sqrt{mn}/\epsilon+n^3}$ time. 
\end{corollary}
As discussed in \cref{subsec:application}, lasso regression can also be embedded into the quantum GLM sparsification by introducing $n$ augmented data points and assigning suitable loss functions $f_i$ to them.
\begin{corollary}[Quantum Lasso Regression]\label{cor:lasso-regression}
Let $\lambda>0$ be a regularization parameter.
There exists a quantum algorithm that, given query access to a matrix $A\in \RR^{m\times n}$ (with row sparsity $r \leq n$) and vector $b\in \RR^m$, and $\epsilon> 0$, outputs with high probability an $x\in \RR^n$ such that $\norms{Ax-b}_2^2 +\lambda \norm{x}_1\leq \parens{1+\epsilon} \min _{x\in\RR^n }\paren{\norms{Ax-b}_2^2+\norms{x}_1}$, in $\widetilde O \parens{r\sqrt{mn}/\epsilon+n^3/\epsilon^2}$ time. 
\end{corollary}

\end{document}